\documentclass[12pt]{article}
\usepackage[right=1in,left=1in,top=1in,bottom=1in]{geometry}
\usepackage{hyperref}
\hypersetup{colorlinks, citecolor=blue, filecolor=blue, linkcolor=blue, urlcolor=blue}
\usepackage{graphicx}
\usepackage[round]{natbib}
\usepackage{amsmath,amsthm}
\usepackage{algorithm}
\usepackage{algpseudocode}
\usepackage{float}
\usepackage{amssymb}
\usepackage{caption}

\newtheorem{theorem}{Theorem}

\newtheorem{assumption}{Assumption}
\newtheorem{proposition}{Proposition}
\newtheorem{lemma}{Lemma}
\newtheorem{remark}{Remark}
\DeclareMathOperator{\Var}{Var}

\usepackage{setspace}
\usepackage{sectsty}
\sectionfont{\large}
\subsectionfont{\normalsize}
\subsubsectionfont{\normalsize}

\title{ \vspace*{-2.5cm} A Variance-Based Test for Heterogeneous Treatment Effects}

\author{Fangzhou Yu\thanks{School of Economics, University of Sydney.
\href{mailto:fangzhou.yu@sydney.edu.au}{fangzhou.yu@sydney.edu.au}}}

\date{ \vspace*{0.5cm} February, 2026} 

\begin{document}

\bgroup
\let\footnoterule\relax

\begin{singlespace}
  \maketitle

  \begin{abstract}

    This paper proposes a robust nonparametric hypothesis test for the existence of heterogeneous treatment effects. We focus on the variance of the Conditional Average Treatment Effect (CATE) as a natural omnibus parameter, where a non-zero variance implies the presence of relevant heterogeneity. Standard inference for this parameter faces a fundamental theoretical challenge. On one hand, evaluating variance components on the same sample leads to null degeneracy, where the asymptotic variance collapses to zero under the null hypothesis of homogeneity, invalidating standard Gaussian inference. On the other hand, decoupling the empirical processes via standard sample-splitting breaks the Neyman orthogonality of the doubly robust scores due to their nonlinear squared loss, which prevents the cancellation of first-order regularization biases. To resolve this challenge, we propose a novel Intra-Fold Sample-Splitting algorithm. By evaluating variance components on mutually disjoint subsamples while coupling them to identical out-of-fold nuisance estimators, our procedure achieves algebraic cancellation of the nuisance biases. We prove this restores consistency and asymptotic normality, and ensures Type I error control. Monte Carlo simulations demonstrate that the proposed test achieves superior size control relative to existing tests while maintaining high power. In an empirical application to the NSW job training program, the test detects significant heterogeneity that traditional nonparametric tests fail to uncover.

  \end{abstract}
\end{singlespace}
\thispagestyle{empty}

\clearpage
\egroup
\setcounter{page}{1}


\section{Introduction  \label{sec:intro}}

The analysis of causal effects has traditionally centered on the Average Treatment Effect (ATE), which summarizes the mean impact of a policy or intervention across an entire population. While nonparametric estimators for the ATE with valid statistical inference are now well-established under the unconfoundedness assumption \citep[e.g.,][]{robins1994estimation,chernozhukov2018double}, the ATE often masks substantial heterogeneity in individual responses. Recognizing this heterogeneity is crucial for understanding the underlying mechanisms of a treatment and for designing optimal policies that tailor interventions to specific subpopulations \citep{heckman1997making, athey2017state}.

Recent advances in causal machine learning have facilitated flexible estimation of the Conditional Average Treatment Effect (CATE) function, \(\tau(x) = \mathbb{E}[Y(1) - Y(0) | X=x]\), even in high-dimensional settings \citep[e.g.,][]{wager2018estimation, nie2021quasi}. However, obtaining valid statistical inference for the full CATE function remains a formidable challenge. The complexity of modern machine learning algorithms often precludes the use of classical empirical process theory, and the regularization bias required for estimation makes formally testing hypotheses about the shape of \(\tau(x)\) difficult. Consequently, researchers often face a trade-off between the robust inference available for the ATE and the granular, yet often unstable, characterization of the full CATE curve.

To bridge this gap, we propose a robust hypothesis test for the existence of heterogeneous treatment effects. Rather than attempting to estimate the shape of the heterogeneity immediately, we ask a preliminary question: Is the treatment effect constant across subpopulations defined by covariates? We answer this by conducting inference on a single scalar parameter, the variance of the CATE, \(\theta_0 = \Var(\tau(X))\). This parameter serves as a natural omnibus measure. If \(\theta_0 = 0\), the effects are homogeneous almost surely, and if \(\theta_0 > 0\), relevant heterogeneity exists, justifying further granular investigation.

Developing a valid test for \(\theta_0\) presents a theoretical challenge at the intersection of causal inference and machine learning. To robustly estimate this variance, one must rely on the doubly robust pseudo-outcome, which is the influence function of the ATE. However, testing the null hypothesis of homogeneity using these pseudo-outcomes introduces an impasse characterized by two problems, null degeneracy and the breakdown of Neyman orthogonality.

First, under the null hypothesis, the true parameter lies on the boundary of the parameter space, and the CATE collapses to the ATE. In this boundary case, the true variance components of the pseudo-outcome become identical, and the influence function of the standard variance difference degenerates to zero almost surely. If the test statistic is computed on a single full sample, the asymptotic variance collapses and standard Gaussian inference breaks down. The modern semiparametric resolution to such degeneracy is to decouple the empirical processes of the variance components via sample-splitting \citep{williamson2023general}.

Second, resolving null degeneracy via standard sample-splitting inadvertently causes a breakdown of Neyman orthogonality, which is a key ingredient for developing \(\sqrt{n}\)-consistent semiparametric estimators \citep{chernozhukov2018double}. By the Law of Total Variance, \(\theta_0\) is identified as the difference between the total and residual variances of the pseudo-outcome. While the pseudo-outcome itself is doubly robust, its squared loss is not, introducing \(O_P(n^{-1/4})\) first-order regularization biases into both variance components. These biases can cancel each other when both components are constructed using the exact same nuisance estimators. Standard splitting destroys this symmetry. By evaluating the components on separate folds to decouple their empirical processes, it forces the use of independently trained machine learning models, and thus, their biases fail to cancel. When the test statistic is scaled by \(\sqrt{n}\), this uncancelled residual error diverges to infinity, invalidating asymptotic inference.

To resolve this methodological impasse, a valid test must evaluate the variance components on disjoint observations to resolve degeneracy, while applying identical nuisance estimators to preserve bias cancellation. We achieve this via a novel Intra-Fold Sample-Splitting algorithm. We partition the data into \(K\) main folds and train a single set of nuisance functions out-of-fold. We randomly bisect each in-fold dataset into two mutually disjoint halves and compute the total variance exclusively on the first half and the residual variance on the second, applying the exact same out-of-fold nuisance estimators to both. This paired structure ensures a strictly positive asymptotic variance under the null and cancels out the non-orthogonal squared biases. We formally prove this procedure yields consistent and asymptotically normal estimators, and guarantees valid Type I error control on the boundary.

Our work contributes to the growing literature on testing for treatment effect heterogeneity. Existing methods largely fall into two categories, projection-based tests and distributional tests. A prominent strand of literature focuses on testing whether the projection of the CATE onto a specific set of basis functions of covariates is zero. \citet{crump2008nonparametric} propose a nonparametric test based on sieve estimation, while \citet{semenova2021debiased} develop a Double/Debiased Machine Learning (DML) inference framework for the coefficients of a linear projection of the CATE. Theoretically, these projection-based methods are consistent against general nonlinear alternatives provided the number of basis functions grows sufficiently with the sample size. However, in practice, this approach faces a fundamental trade-off between approximation error and statistical power. Testing the joint significance of a high-dimensional vector of coefficients consumes degrees of freedom, diluting statistical power. Conversely, specifying a parsimonious basis to maximize power risks inconsistency if the true heterogeneity is orthogonal to the chosen subspace. In contrast, our CATE Variance Test targets a single scalar parameter. Because \(\theta_0 = 0\) is a necessary and sufficient condition for a constant CATE, our test remains consistent against any deviation from the null without incurring the power penalty associated with high-dimensional coefficient testing.

A second strand of literature focuses on distributional effects, testing for differences in the marginal distributions or variances of potential outcomes \citep[e.g.,][]{ding2016randomization, chung2021permutation}. While observing a difference in marginal distributions of potential outcomes implies the existence of individual treatment effect heterogeneity, it is not a direct test of moderation by observables. It is possible for individual effects to vary while the conditional average effect \(\tau(x)\) remains constant. Our test specifically isolates the heterogeneity explained by covariates, making it directly relevant for policy evaluation and design. Beyond testing for treatment effect heterogeneity, our algorithm offers a generalizable framework for conducting valid hypothesis testing on nonlinear transformations of doubly robust scores.

The remainder of the paper is organized as follows. Section \ref{sec:framework} establishes the econometric framework and the identification of the target parameter via pseudo-outcomes. Section \ref{sec:tests} formalizes the theoretical tension between null degeneracy and Neyman orthogonality, introduces our Intra-Fold Sample-Split algorithm and establishes its asymptotic properties. Section \ref{sec:sim} presents Monte Carlo simulation results comparing our test to existing projection-based alternatives. Section \ref{sec:app} applies the test to empirical data from the NSW job training program, and Section \ref{sec:conclude} concludes. All proofs are collected in the Appendix.

\section{Framework and Identification \label{sec:framework}}

In this section, we define the causal parameters of interest, state the assumptions required for identification and inference, and derive the variance decomposition that forms the basis of our test statistic.

\subsection{Setup}

We follow the potential outcomes framework \citep{rubin1974estimating}. We observe a random sample of \(n\) independent and identically distributed units \(O_i = (Y_i, D_i, X_i)\) for \(i=1, \dots, n\), drawn from an unknown distribution \(P_0\). Here, \(D_i \in \{0, 1\}\) is a binary treatment indicator, \(X_i \in \mathcal{X} \subset \mathbb{R}^p\) is a vector of covariates, and \(Y_i \in \mathbb{R}\) is the observed outcome. Let \(Y_i(1)\) and \(Y_i(0)\) denote the potential outcomes under treatment and control, respectively. The observed outcome relates to the potential outcomes via the consistency condition \(Y_i = D_i Y_i(1) + (1-D_i)Y_i(0)\). The fundamental problem of causal inference is that for any unit \(i\), we observe only one of the two potential outcomes. Our primary focus is the CATE, defined as
\[
  \tau_0(x) = \mathbb{E}[Y_i(1) - Y_i(0) | X_i = x].
\]
We also define the ATE, denoted by \(\tau_{\mathrm{ATE}} = \mathbb{E}[\tau_0(X_i)]\). To facilitate identification, we define the nuisance functions \(\mu_0(d, x) = \mathbb{E}[Y_i | D_i=d, X_i=x]\) for \(d \in \{0, 1\}\) representing the conditional outcome means, and \(e_0(x) = P(D_i=1 | X_i=x)\) representing the propensity score.

We invoke the standard assumptions for causal identification in observational studies \citep{rosenbaum1983central}, alongside regularity conditions required for valid asymptotic inference.

\begin{assumption}[Unconfoundedness]
  \label{ass:unconfound}
  Conditional on covariates \(X_i\), the treatment assignment is independent of potential outcomes
  \[
    D_i \perp (Y_i(1), Y_i(0)) \mid X_i.
  \]
\end{assumption}

\begin{assumption}[Overlap]
  \label{ass:overlap}
  The propensity score is strictly bounded away from zero and one. There exists a constant \(\xi > 0\) such that
  \[
    \xi \leq e_0(x) \leq 1-\xi
  \]
  almost surely for all \(x \in \mathcal{X}\).
\end{assumption}

\begin{assumption}
  \label{ass:nondegen}
  (i) The outcome \(Y_i\) has bounded fourth moments: \(\mathbb{E}[Y_i^4] < \infty\). (ii) Non-degeneracy: the variance of the squared centered pseudo-outcome is strictly bounded away from zero. There exists a constant \(c > 0\) such that
  \[
    \Var\!\left( (\psi_0(O_i) - \tau_{\mathrm{ATE}})^2 \right) > c,
  \]
  where \(\psi_0(O_i)\) is the doubly robust pseudo-outcome defined in Equation~\eqref{eq:pseudo}.
\end{assumption}

Assumptions \ref{ass:unconfound} and \ref{ass:overlap} allow for the identification of the CATE function \(\tau_0(x) = \mu_0(1, x) - \mu_0(0, x)\). Assumption \ref{ass:nondegen}(i) ensures finite moments needed for the central limit theorem and for the empirical-process arguments underlying double machine learning. Condition (ii) directly guarantees that the asymptotic variance of our test statistic is bounded away from zero on the boundary of the parameter space, eliminating the pathological degeneracy that would otherwise arise under the null hypothesis of homogeneity.

\subsection{The Target Parameter and Hypotheses}

We investigate whether the treatment effect is constant across the population defined by \(X\). Formally, we define the CATE Variance parameter:
\[
  \theta_0 = \Var(\tau_0(X_i)).
\]
The variance serves as an omnibus measure of heterogeneity. If \(\theta_0 = 0\), the CATE is constant almost surely (i.e., \(\tau_0(x) = \tau_{\mathrm{ATE}}\) for all \(x\)). If \(\theta_0 > 0\), there exists variation in the treatment effect explained by the covariates. Accordingly, we test the null hypothesis of homogeneity against the one-sided alternative of heterogeneity:
\begin{equation}
  \label{eq:hypothesis}
  H_0: \theta_0 = 0 \quad \text{vs.} \quad H_1: \theta_0 > 0.
\end{equation}
Since \(\theta_0\) is non-negative, the null hypothesis lies on the boundary of the parameter space. We address the inferential implications of this boundary condition in Section \ref{sec:tests}.

\subsection{Identification via Pseudo-Outcomes}

A direct estimator of \(\Var(\tau_0(X))\) based on a plug-in estimate of the function \(\hat{\tau}(x)\) would suffer from first-order regularization bias, particularly when \(X\) is high-dimensional. To address this, we utilize a doubly robust pseudo-outcome, also known as the uncentered influence function for the ATE. Define the pseudo-outcome \(\psi(O_i)\) as
\begin{equation}
  \label{eq:pseudo}
  \psi(O_i) = \mu_0(1, X_i) - \mu_0(0, X_i) + \frac{D_i(Y_i - \mu_0(1, X_i))}{e_0(X_i)} - \frac{(1-D_i)(Y_i - \mu_0(0, X_i))}{1-e_0(X_i)}.
\end{equation}
This pseudo-outcome possesses two critical properties. First, it is an unbiased signal of the CATE
\[
  \mathbb{E}[\psi(O_i) | X_i] = \tau_0(X_i),
\]
which also implies \(\mathbb{E}[\psi(O_i)] = \tau_{\mathrm{ATE}}\). Second, it allows us to identify \(\theta_0\) through a variance decomposition. By the Law of Total Variance applied to \(\psi(O_i)\), we have
\[
  \Var(\psi(O_i)) = \Var(\mathbb{E}[\psi(O_i)|X_i]) + \mathbb{E}[\Var(\psi(O_i)|X_i)].
\]
Substituting the conditional expectation with \(\tau_0(X_i)\), we can rearrange this to identify the CATE variance
\[
  \theta_0 = \Var(\psi(O_i)) - \mathbb{E}[(\psi(O_i) - \tau_0(X_i))^2].
\]
Or, expressed in terms of Mean Squared Error which facilitates our estimation strategy
\begin{equation}
  \label{eq:cate_var}
  \theta_0 = \underbrace{\mathbb{E}[(\psi(O_i) - \tau_{\mathrm{ATE}})^2]}_{V_{\mathrm{tot}}} - \underbrace{\mathbb{E}[(\psi(O_i) - \tau_0(X_i))^2]}_{V_{\mathrm{res}}}.
\end{equation}

Equation \eqref{eq:cate_var} provides the identification result for our test. It expresses the CATE variance as the difference between \(V_{\mathrm{tot}}\), the MSE of the best constant predictor of the pseudo-outcome (\(\tau_{\mathrm{ATE}}\)), and \(V_{\mathrm{res}}\), the MSE of the best conditional predictor (\(\tau_0(X)\)).

\section{Test for Heterogeneous Treatment Effect\label{sec:tests}}

In this section, we develop a formal hypothesis test for the presence of heterogeneous treatment effects. Having identified the CATE variance, \(\theta_0 = \Var(\tau_0(X))\), as our target parameter in Section \ref{sec:framework}, we test the hypotheses in Equation \eqref{eq:hypothesis}. This test leverages the identification result derived in Equation \eqref{eq:cate_var}. While estimating \(\theta_0\) fits within the general framework of semiparametric inference, the null hypothesis poses a unique theoretical challenge known as null-degeneracy. Below, we derive the influence function for \(\theta_0\), analyze its properties, and detail the algorithm for the hypothesis test.

We first derive the influence function for \(\theta_0\).
\begin{proposition}[Influence Function for \(\theta_0\)]
  \label{prp:cvtif}
  Under Assumptions~\ref{ass:unconfound}--\ref{ass:nondegen}, the influence function for the CATE variance \(\theta_0\) is given by
  \[
    \phi_{\theta}(O_i) = (\psi(O_i) - \tau_{\mathrm{ATE}})^2 - (\psi(O_i) - \tau_0(X_i))^2 - \theta_0,
  \]
  where \(\psi(O_i)\) is the pseudo-outcome defined in Equation \eqref{eq:pseudo}.
\end{proposition}

Proof. See Appendix.

Based on Proposition \ref{prp:cvtif}, a standard "one-step" efficient estimator can be constructed by solving the empirical equation \(n^{-1} \sum \hat\phi_{\theta}(O_i) = 0\), where \(\hat\phi_{\theta}(O_i)\) is obtained by plugging in the estimated nuisance parameters. Under the alternative hypothesis, \(H_1: \theta_0 > 0\), standard semiparametric theory guarantees that such an estimator is \(\sqrt{n}\)-consistent and asymptotically normal
\[
  \sqrt{n}(\hat\theta - \theta_0) \xrightarrow{d} \mathcal{N}(0, \Var(\phi_\theta)),
\]
provided that the nuisance parameters converge at sufficiently fast rates, typically \(n^{-1/4}\) \citep{chernozhukov2018double}. This allows us to employ modern machine learning methods to estimate \(\mu_0\) and \(e_0\), and plug in the pseudo-outcome \(\psi(O_i)\). Constructing \(\hat\theta\) also requires feasible estimators for \(\tau_{\mathrm{ATE}}\) and \(\tau_0(X)\), which we review in the next section.

\subsection{Estimation of ATE and CATE}

Existing strategies for estimating \(\tau_0(x)\) and \(\tau_{\mathrm{ATE}}\) largely fall into two categories, the T-learner and the DR-learner. The T-learner estimates the conditional means \(\mu_0(1, x)\) and \(\mu_0(0, x)\) and computes their differences to obtain an estimator for the treatment effect. For the ATE, a T-learner is \(n^{-1}\sum_{i=1}^n (\hat\mu(1, X_i) - \hat\mu(0, X_i))\), and for the CATE, a T-learner is \(\hat\mu(1, X_i) - \hat\mu(0, X_i)\). By the triangle inequality, the \(L_2\) error of the T-learner is bounded by the errors of the baseline outcome models. Therefore, provided the nuisance estimators \(\hat{\mu}(1, \cdot)\) and \(\hat{\mu}(0, \cdot)\) satisfy the \(o_P(n^{-1/4})\) rate, we can show that the T-learner also satisfies this rate and can be applied in our algorithm. However, as noted by \citet{kunzel2019metalearners}, T-learners can suffer from regularization bias, particularly when the CATE function is sparser than the baseline outcome functions or when there is poor overlap between treatment groups.

The DR-learner, on the other hand, treats ATE and CATE estimation as a direct regression of the pseudo-outcome \(\psi(O_i)\) on the covariates. The DR-learner for the ATE is \(n^{-1}\sum_{i=1}^n \hat\psi(O_i)\), which is also known as the Augmented Inverse Propensity Weighting (AIPW) estimator. Because the pseudo-outcome is Neyman orthogonal, \(\hat{\tau}_{\mathrm{ATE}}\) is \(\sqrt{n}\)-consistent provided the product of the \(L_2\) estimation errors for the propensity score and outcome mean vanishes at an \(o_P(n^{-1/2})\) rate \citep[e.g.,][]{robins1994estimation, chernozhukov2018double}. For the CATE, it is the minimizer of the mean squared error \(\sum_{i=1}^n (\hat{\psi}(O_i) - f(X_i))^2\). Its estimation error is bounded by the oracle smoothing error of the CATE plus the product of the nuisance errors \(\|\hat{e} - e_0\|_{P,2} \times \|\hat{\mu} - \mu_0\|_{P,2}\) \citep{kennedy2023towards}, where \(\|\cdot\|_{P,2}\) denotes the \(L_2(P_0)\) norm. This imparts a "double robustness of rates." Even if the baseline outcome model \(\hat{\mu}\) converges at a rate slower than \(n^{-1/4}\) due to complex confounding, the DR-learner can still achieve the requisite \(o_P(n^{-1/4})\) rate, provided the propensity score converges sufficiently fast and the true CATE is sufficiently smooth.

In this paper, we adopt the DR-learner for both ATE and CATE estimation. It often yields more stable estimates than differencing two regression functions, and the convergence rate depends on the product of nuisance errors, making it robust to misspecification of the nuisance models. However, simply plugging the nuisance estimators into a standard full-sample or cross-fitting empirical analogue of \(\theta_0\) fails to yield valid inference. We formalize this fundamental breakdown of Neyman orthogonality in the next section.

\subsection{Null Degeneracy and the Breakdown of Orthogonality\label{sec:impasse}}

To develop a valid semiparametric test based on the variance of CATE, the first hurdle is the problem of null degeneracy. Under the null hypothesis of homogeneity \(H_0: \theta_0 = 0\), the true CATE is constant almost surely, i.e., \(\tau_0(X_i) = \tau_{\mathrm{ATE}}\). Consequently, the true total and residual losses are identical, and their corresponding influence functions coincide perfectly. If one computes the empirical analogues of \(V_{\mathrm{tot}}\) and \(V_{\mathrm{res}}\) using the same sample of observations, the empirical processes become perfectly correlated, and the asymptotic variance of their difference collapses to zero. This degeneracy violates the regularity conditions required for standard Gaussian approximations and destroys the size calibration of the test: rather than attaining its nominal level, the same-sample statistic degenerates and becomes severely conservative (Appendix \ref{sec:appnaive}). \citet{williamson2023general} suggest that this degeneracy can be resolved by evaluating the components on disjoint subsets of the data via sample-splitting.

However, resolving null degeneracy via standard sample-splitting breaks down the Neyman orthogonality of the influence function in Proposition \ref{prp:cvtif}. To formalize this, consider the pathwise Gâteaux derivative of the expected squared residual loss, \(\mathbb{E}[(\psi - \tau_0(X))^2]\), with respect to the propensity score \(e(x)\). The expected first-order bias depends on the cross-term conditional on \(X_i\),
\[
  \mathbb{E}\left[ 2(\psi(O_i) - \tau_0(X_i)) \frac{\partial \psi}{\partial e}(O_i) \bigg| X_i \right].
\]
Substituting the residual error
\[
  \psi(O_i) - \tau_0(X_i) = \frac{D_i(Y_i-\mu_0(1, X_i))}{e_0(X_i)} - \frac{(1-D_i)(Y_i-\mu_0(0, X_i))}{1-e_0(X_i)}
\]
and its partial derivative 
\[
  \frac{\partial \psi}{\partial e}(O_i) = - \frac{D_i(Y_i-\mu_0(1, X_i))}{e_0(X_i)^2} - \frac{(1-D_i)(Y_i-\mu_0(0, X_i))}{(1-e_0(X_i))^2},
\]
the cross-products strictly vanish since the treatment indicator satisfies \(D_i(1-D_i) = 0\). Using the unconfoundedness assumption to replace the expected squared residual outcomes with the true conditional variances \(\sigma_1^2(X_i)\) and \(\sigma_0^2(X_i)\),
\[ 
  \mathbb{E}\left[ 2(\psi(O_i) - \tau_0(X_i)) \frac{\partial \psi}{\partial e}(O_i) \bigg| X_i \right] = 2 \left( - \frac{\sigma_1^2(X_i)}{e_0(X_i)^2} + \frac{\sigma_0^2(X_i)}{(1-e_0(X_i))^2} \right) \equiv g(X_i).
\]
Crucially, this derivative \(g(X_i)\) is generally non-zero. Because this derivative does not vanish, the squared pseudo-outcome is not Neyman orthogonal. Any plug-in estimator for the residual variance \(V_{\mathrm{res}}\) is therefore contaminated by a first-order regularization bias of order \(O_P(\|\hat{e} - e_0\|_{P,2})\). The estimator for the total variance, \(V_{\mathrm{tot}} = \mathbb{E}[(\psi - \tau_{\mathrm{ATE}})^2]\), suffers from the same non-orthogonal bias \(g(X_i)\).

The target parameter \(\theta_0\) remains \(\sqrt{n}\)-consistent only because of an exact algebraic cancellation. If both variance components are evaluated using the exact same nuisance estimators, their respective first-order biases \(g(X_i)\) are mathematically identical and cancel one another when taking the difference \(V_{\mathrm{tot}} - V_{\mathrm{res}}\). Standard sample-splitting, which evaluates the two variance components on different data folds, structurally destroys this delicate symmetry by forcing the use of independently trained machine learning nuisance estimators (e.g., evaluating \(V_{\mathrm{tot}}\) with an out-of-fold propensity score \(\hat{e}_{\mathrm{odd}}\) and \(V_{\mathrm{res}}\) with \(\hat{e}_{\mathrm{even}}\)). Because these independently trained nuisance estimators differ in finite samples, their induced non-orthogonal biases no longer match. The uncancelled first-order bias in the split-sample estimator becomes approximately 
\[
  \text{Bias}(\hat{\theta}_{\mathrm{split}}) \approx \int g(X) \big( \hat{e}_{\mathrm{odd}}(X) - \hat{e}_{\mathrm{even}}(X) \big) dP_0(X).
\]
Standard rates only guarantee that independently trained nuisance estimators differ by \(o_P(n^{-1/4})\), so the \(\sqrt{n}\)-scaled test statistic inherits a residual bias of order \(o_P(n^{1/4})\) — a quantity that need not converge to zero, invalidating asymptotic inference.

\subsection{The Intra-Fold Sample-Split Algorithm}

To resolve the methodological impasse formalized in Section \ref{sec:impasse}, a valid testing procedure must simultaneously evaluate the total and residual variance components on strictly disjoint sets of observations and construct these components using the same nuisance estimators to preserve the algebraic cancellation of the squared pseudo-outcomes, thereby restoring Neyman orthogonality.

We achieve these requirements via a novel Intra-Fold Sample-Split (IF-SS) algorithm, detailed in Algorithm \ref{alg:ss_cvt}. Instead of splitting the evaluation of the variance components across entirely different main folds, our algorithm introduces an internal data partition. We first partition the data into \(K\) main folds and train a single set of nuisance functions on the out-of-fold data. We randomly bisect each in-fold evaluation dataset into two mutually disjoint halves. We compute the total variance exclusively on the first half and the residual variance exclusively on the second half, applying the identically trained out-of-fold nuisance estimators to both.

\begin{algorithm}
  \caption{Intra-Fold Sample-Split CATE Variance Test (IF-SS-CVT)}
  \label{alg:ss_cvt}
  \begin{algorithmic}[1]
    \Require Data \(\{(Y_i, D_i, X_i)\}_{i=1}^n\), Number of folds \(K \ge 2\), Significance level \(\alpha\).
    \Ensure Test statistic \(Z_\theta\) and rejection decision.
    \Statex \textbf{1. Partitioning}
    \State Randomly partition indices \(\{1, \dots, n\}\) into \(K\) disjoint main folds \(\mathcal{I}_1, \dots, \mathcal{I}_K\). Let \(n_k = |\mathcal{I}_k|\).
    \State For each fold \(k\), randomly sub-split the evaluation fold \(\mathcal{I}_k\) into two mutually disjoint subsets, \(\mathcal{I}_{k,\mathrm{tot}}\), and \(\mathcal{I}_{k,\mathrm{res}}\), of equal size \(m_k = n_k / 2\).
    \Statex \textbf{2. Nuisance Training and Paired Evaluation}
    \For{\(k = 1, \dots, K\)}
    \State Train estimators \(\hat{\eta}_k = (\hat{e}_k, \hat{\mu}_k)\) and learners \(\hat{\tau}_k\), \(\hat{\tau}_{\mathrm{ATE},k}\) on \(\mathcal{I}_{-k} = \{1, \dots, n\} \setminus \mathcal{I}_k\).
    \For{each unit \(i \in \mathcal{I}_k\)}
    \State \(\hat{\psi}_{i,k} \gets \hat{\mu}_k(1, X_i) - \hat{\mu}_k(0, X_i) + \frac{D_i(Y_i - \hat{\mu}_k(1, X_i))}{\hat{e}_k(X_i)} - \frac{(1-D_i)(Y_i - \hat{\mu}_k(0, X_i))}{1-\hat{e}_k(X_i)}\)
    \EndFor
    \State \textit{Evaluate Total Variance strictly on \(\mathcal{I}_{k,\mathrm{tot}}\):}
    \State \(\hat{V}_{\mathrm{tot},k} \gets \frac{1}{m_k} \sum_{i \in \mathcal{I}_{k,\mathrm{tot}}} (\hat{\psi}_{i,k} - \hat{\tau}_{\mathrm{ATE},k})^2\)
    \State \(\hat{\sigma}^2_{\mathrm{tot},k} \gets \frac{1}{m_k - 1} \sum_{i \in \mathcal{I}_{k,\mathrm{tot}}} \left( (\hat{\psi}_{i,k} - \hat{\tau}_{\mathrm{ATE},k})^2 - \hat{V}_{\mathrm{tot},k} \right)^2\)
    \State \textit{Evaluate Residual Variance strictly on \(\mathcal{I}_{k,\mathrm{res}}\), using the same \(\hat{\psi}_{i,k}\):}
    \State \(\hat{V}_{\mathrm{res},k} \gets \frac{1}{m_k} \sum_{i \in \mathcal{I}_{k,\mathrm{res}}} (\hat{\psi}_{i,k} - \hat{\tau}_k(X_i))^2\)
    \State \(\hat{\sigma}^2_{\mathrm{res},k} \gets \frac{1}{m_k - 1} \sum_{i \in \mathcal{I}_{k,\mathrm{res}}} \left( (\hat{\psi}_{i,k} - \hat{\tau}_k(X_i))^2 - \hat{V}_{\mathrm{res},k} \right)^2\)
    \EndFor
    \Statex \textbf{3. Aggregation and Inference}
    \State Compute split-sample variance estimate: \(\hat{\theta}_{\mathrm{split}} \gets \frac{1}{K} \sum_{k=1}^K (\hat{V}_{\mathrm{tot},k} - \hat{V}_{\mathrm{res},k})\)
    \State Compute standard error: \(\widehat{SE} \gets \sqrt{ \frac{1}{K^2} \sum_{k=1}^K \left( \frac{\hat{\sigma}^2_{\mathrm{tot},k}}{m_k} + \frac{\hat{\sigma}^2_{\mathrm{res},k}}{m_k} \right) }\)
    \State Compute standardized test statistic: \(Z_\theta \gets \hat{\theta}_{\mathrm{split}} / \widehat{SE}\)
    \State \Return \textbf{Reject} \(H_0\) if \(Z_\theta > z_{1-\alpha}\), otherwise \textbf{Fail to reject}.
  \end{algorithmic}
\end{algorithm}

To establish the asymptotic validity of Algorithm \ref{alg:ss_cvt}, we impose regularity conditions on the estimators used. We maintain Assumptions~\ref{ass:unconfound}--\ref{ass:nondegen} from Section \ref{sec:framework} and further introduce the following regularity conditions on nuisance estimators.

\begin{assumption}
  \label{ass:nuisance_regularity}
  (i) Convergence Rates:
  \begin{align*}
    \|\hat{e}_k - e_0\|_{P,2} = o_P(n^{-1/4}) \quad \text{and} \quad \|\hat{\mu}_k - \mu_0\|_{P,2} = o_P(n^{-1/4}), \\
    \|\hat{\tau}_k - \tau_0\|_{P,2} = o_P(n^{-1/4}) \quad \text{and} \quad |\hat{\tau}_{\mathrm{ATE},k} - \tau_{\mathrm{ATE}}| = O_P(n^{-1/2}).
  \end{align*}

  (ii) Uniform Boundedness: There exist constants \(\xi > 0\) and \(C < \infty\) such that with probability approaching 1, \(\hat{e}_k(X) \in [\xi, 1-\xi]\) and \(\max\!\big(|\hat{\mu}_k(d, X)|, |\hat{\tau}_k(X)|, |\mu_0(d, X)|, |\tau_0(X)|\big) \le C\) almost surely.
\end{assumption}

Because Algorithm \ref{alg:ss_cvt} algebraically cancels the non-orthogonal bias, the only remaining estimation errors depend strictly on the doubly robust linear pseudo-outcome terms and the Mean Squared Error of the CATE estimator itself (\(\|\hat{\tau}_k - \tau_0\|_{P,2}^2\)). Provided Assumption \ref{ass:nuisance_regularity} holds, these remaining errors rigorously vanish at an \(o_P(n^{-1/2})\) rate. We formalize the asymptotic validity of this test in Theorem \ref{thm:ss_cvt}.

\begin{theorem}[Asymptotic Validity of IF-SS-CVT]
  \label{thm:ss_cvt}
  Suppose Assumptions~\ref{ass:unconfound}--\ref{ass:nuisance_regularity} hold. Let \(Z_{\theta}\) be the standardized test statistic computed via Algorithm \ref{alg:ss_cvt} with a fixed number of folds \(K \ge 2\). As \(n \rightarrow \infty\), under both the null hypothesis \(H_0: \theta_0 = 0\) and the alternative \(H_1: \theta_0 > 0\), the standardized estimator converges to a standard normal distribution
  \[
    \frac{\hat{\theta}_{\mathrm{split}} - \theta_0}{\widehat{SE}} \xrightarrow{d} \mathcal{N}(0, 1).
  \]
  Consequently, under the null hypothesis \(H_0: \theta_0 = 0\), the test controls the Type I error rate at level \(\alpha\)
  \[
    \lim_{n \to \infty} P(Z_\theta > z_{1-\alpha} \mid H_0) = \alpha.
  \]
  Under the alternative hypothesis \(H_1: \theta_0 > 0\), the test is consistent against any fixed alternative
  \[
    \lim_{n \to \infty} P(Z_\theta > z_{1-\alpha} \mid H_1) = 1.
  \]
\end{theorem}

Proof. See Appendix.

\section{Simulation \label{sec:sim}}

We evaluate the performance of the proposed test using Monte Carlo simulations. In all designs, we generate \(n \in \{250, 500, 1000, 2000\}\) independent and identically distributed observations \(O_i = (Y_i, D_i, X_i)\) where \(X_i \in \mathbb{R}^p\). The outcome follows a common structural model
\[
  Y_i = \mu_0(X_i) + D_i \cdot \tau(X_i) + \varepsilon_i,
\]
where \(\mu_0(x)\) is the baseline outcome function, \(\tau(x)\) is the CATE, and \(\varepsilon_i \sim N(0, 1)\). The covariates are drawn from a multivariate normal distribution \(X_i \sim N(0, \Sigma)\). The treatment assignment \(D_i\) follows a Bernoulli distribution conditional on \(X_i\) with propensity score \(e(x) = (1 + \exp(-x'\alpha))^{-1}\).

We adopt a sparse setting with \(p = 50\) and uncorrelated covariates, \(\Sigma = I_p\). The propensity score depends on the first three covariates, with \(\alpha = (0.2, 0.2, 0.2, 0, \dots, 0)' \in \mathbb{R}^p\). The baseline outcome is a sparse linear function of the first five covariates,
\[
  \mu_0(x) = x_1 + 0.5 x_2 + 0.5 x_3 + 0.3 x_4 + 0.3 x_5.
\]

We examine four specifications of the CATE function \(\tau(x) = \mu(1, x) - \mu(0, x)\):

\begin{enumerate}
  \item \textbf{Constant CATE (Null):} The treatment effect is constant, \(\tau(x) = 1\).
        
  \item \textbf{Linear CATE:} The treatment effect is linear in the first two covariates,
        \[
          \tau(x) = 2x_1 + x_2.
        \]
        
  \item \textbf{Kinked CATE:} The treatment effect is piecewise linear with a kink at zero,
        \[
          \tau(x) = 4 \max(x_1, 0) + 2 \max(x_2, 0).
        \]
        
  \item \textbf{Nonlinear CATE:} The treatment effect is a smooth nonlinear function,
        \[
          \tau(x) = 3\Big(\exp\!\big(\tfrac{x_1}{2}\big) + \exp\!\big(\tfrac{x_2}{2}\big) - 2\exp\!\big(\tfrac{1}{8}\big)\Big).
        \]
\end{enumerate}

Figure~\ref{fig:simvis} provides visualizations of the data generating processes through the scatter plots of \(Y_i\) against \(X_{1i}\), alongside the true conditional mean functions \(\mu(1, x_1)\) and \(\mu(0, x_1)\) evaluated at the mean of all other covariates. The models are designed to reflect qualitatively different patterns of treatment effect heterogeneity. The constant CATE model falls under the null of Equation~\eqref{eq:hypothesis}, while the other three models fall under the alternative. The linear and nonlinear models have a zero ATE by construction, so conventional ATE-targeted approaches such as OLS or IPW would fail to detect the existence of treatment effects.

We implement our proposed Algorithm \ref{alg:ss_cvt} using \(K=5\) folds. We estimate the nuisance parameters and the DR-learner for the CATE function using two machine learning algorithms: Lasso and XGBoost\footnote{We use the \texttt{glmnet} R package for Lasso and the \texttt{xgboost} package for XGBoost.}. To demonstrate the theoretical necessity of our IF-SS structure, we introduce a Naive DML benchmark. This benchmark utilizes standard DML cross-fitting but omits our internal sample-splitting step. Specifically, for each fold \(k\), it computes both the total variance \(\hat{V}_{\mathrm{tot},k}\) and the residual variance \(\hat{V}_{\mathrm{res},k}\) on the entire evaluation fold \(\mathcal{I}_k\) using the identically trained nuisance estimators \(\hat{\eta}_k\). The Naive DML benchmark uses XGBoost for nuisance estimation. While this naive approach preserves Neyman orthogonality, it fails to solve the null degeneracy problem. Under the null hypothesis, the influence function \(\phi_i = (\hat{\psi}_i - \hat{\tau}_{\mathrm{ATE}})^2 - (\hat{\psi}_i - \hat{\tau}(X_i))^2\) converges to zero for all \(i\), so that both the point estimate and the estimated standard error degenerate. Because the two variance components are evaluated on the same observations, the flexible CATE learner contributes a spurious dispersion that biases \(\hat{\theta}\) downward, and dividing this negative bias by a standard error of even smaller order drives the standardized statistic to \(-\infty\); the rejection probability of the one-sided test converges to zero (Proposition \ref{prp:naive}). The naive test is therefore severely undersized rather than unreliable in an unpredictable direction. We characterize this conservative degeneracy formally in Appendix \ref{sec:appnaive}.

\begin{figure}[H]
    \begin{center}
    \caption{Sketches of the Data Generating Processes}\label{fig:simvis}
    \includegraphics[width=6.0in]{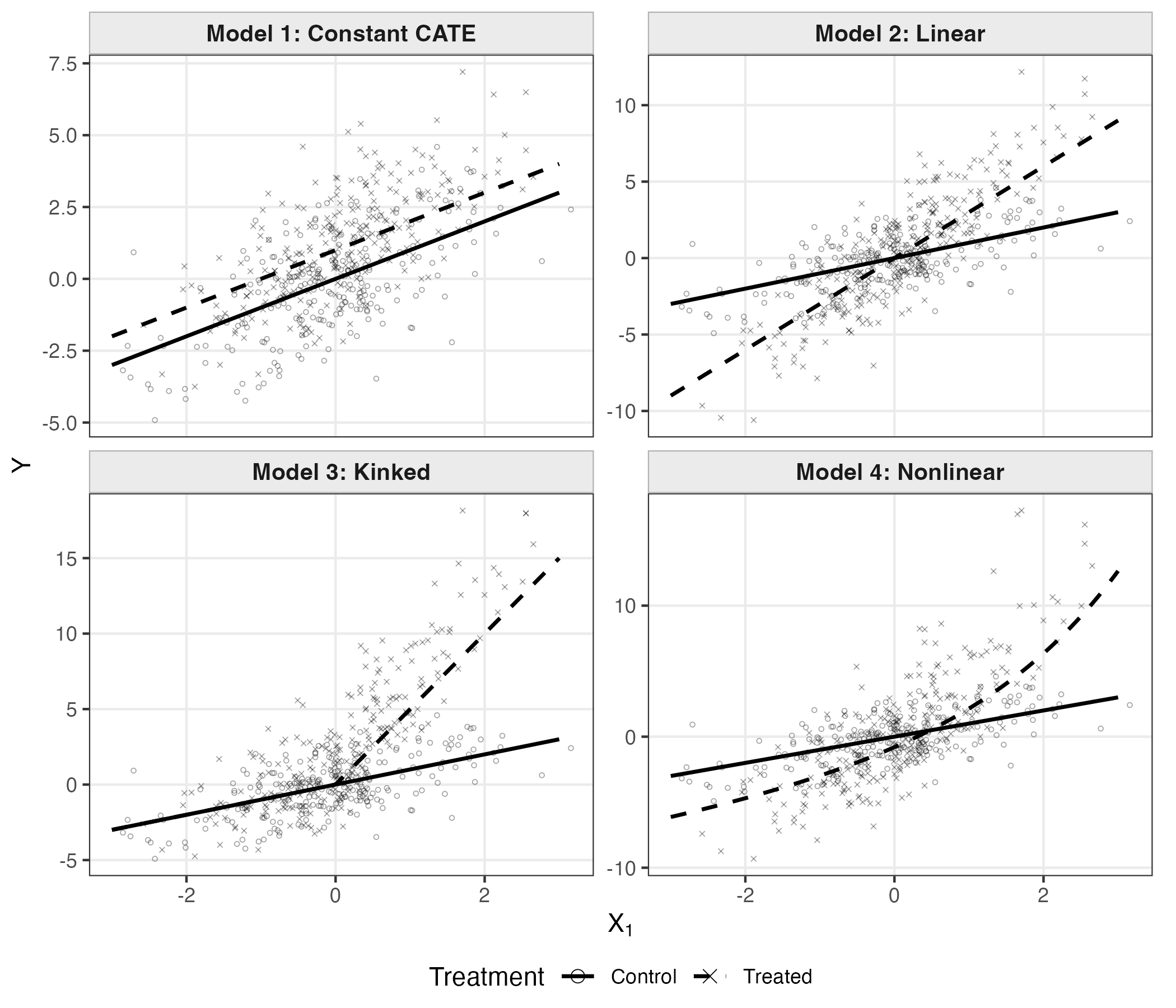}
    \end{center}
    \vspace{0.2cm}
\end{figure}

We also compare the performance with two existing tests in the literature, the nonparametric test proposed by \citet{crump2008nonparametric} (hereinafter CHIM) and the debiased machine learning test proposed by \citet{semenova2021debiased} (hereinafter SC).

We implement the sieve-based nonparametric test proposed by CHIM to evaluate the null hypothesis of a constant conditional average treatment effect. This method approaches the problem by comparing the shapes of the conditional outcome mean functions for the treated and control groups. We approximate these functions, \(\mu_1(x)\) and \(\mu_0(x)\), using a sieve basis expansion \(\mathbf{P}(x) = (1, p_1(x), \dots, p_K(x))'\), where the basis terms are constructed as a linear function of the covariates \(X_i\). This vector includes an intercept and \(K\) covariate-dependent basis terms. We estimate the coefficients by running two separate OLS regressions of the observed outcome \(Y_i\) on \(\mathbf{P}(X_i)\) for the treated and control subsamples, yielding the coefficient vectors \(\hat{\boldsymbol{\xi}}_1 = (\hat{\alpha}_1, \hat{\boldsymbol{\beta}}_1')'\) and \(\hat{\boldsymbol{\xi}}_0 = (\hat{\alpha}_0, \hat{\boldsymbol{\beta}}_0')'\). Under the null hypothesis, the treatment effect is constant, implying that the outcome functions are parallel and their slope coefficients are identical (\(\boldsymbol{\beta}_1 = \boldsymbol{\beta}_0\)). The test statistic evaluates the quadratic distance between these estimated slopes
\[
  T_{\mathrm{Crump}} = (\hat{\boldsymbol{\beta}}_1 - \hat{\boldsymbol{\beta}}_0)' \widehat{\mathbf{V}}_{\beta}^{-1} (\hat{\boldsymbol{\beta}}_1 - \hat{\boldsymbol{\beta}}_0),
\]
where \(\widehat{\mathbf{V}}_{\beta}\) is the robust covariance matrix for the difference in slope estimates.

As a benchmark for high-dimensional settings, we implement the Best Linear Predictor (BLP) test proposed by SC, following Example 2.2 in their paper. This framework approximates the CATE by projecting it onto a linear dictionary of covariates. The core of the method is the construction of a Neyman-orthogonal signal, which is the pseudo-outcome \(\psi(O_i)\) in Equation \eqref{eq:pseudo}, which serves as an unbiased proxy for the latent individual treatment effect. We employ the cross-fitting procedure proposed in their Definition 2.1. The sample is split into \(K\) folds, and for each observation \(i\) in fold \(k\), the signal \(\psi(O_i)\) is constructed using nuisance parameters estimated on the complementary folds. In the second stage, we project this cross-fitted signal onto a vector of covariates \(Z_i\) constructed as a second-order polynomial expansion of \(X_i\) (including interaction terms) to estimate the BLP coefficients. We solve the Lasso optimization problem
\[
  (\hat{\beta}_{0}, \hat{\boldsymbol{\beta}}_{\mathrm{Lasso}}) = \arg\min_{\beta_0, \boldsymbol{\beta}} \frac{1}{n} \sum_{i=1}^n (\hat{\psi}(O_i) - \beta_0 - Z_i'\boldsymbol{\beta})^2 + \lambda \|\boldsymbol{\beta}\|_1.
\]
The null hypothesis of a constant treatment effect implies that the best linear predictor is constant, or equivalently, that the slope coefficients are zero (\(\boldsymbol{\beta} = \mathbf{0}\)). We test this hypothesis using the debiased Lasso estimator to account for regularization bias, constructing a Wald statistic for the joint significance of the slope coefficients.

The empirical rejection proportions over \(1,000\) Monte Carlo replications at the nominal \(\alpha = 0.05\) level are presented in Table \ref{tbl:sim}. Under the constant CATE model, the results demonstrate the impasse detailed in Section \ref{sec:impasse}. The Naive DML estimator is severely undersized under the null: reusing the same evaluation fold biases its point estimate downward while its standard error degenerates even faster, so the standardized statistic drifts to \(-\infty\) and the one-sided test almost never rejects. The drift is slow for regularized learners, which is why the rejection rates remain small but non-zero and essentially flat across the sample sizes considered (see Appendix \ref{sec:appnaive}). Furthermore, CHIM and SC fail severely, with rejection rates far exceeding the nominal level even at large sample sizes. In contrast, our proposed IF-SS-CVT maintains excellent size control across all sample sizes.

Under the alternative hypotheses, all tests show consistent high power when \(n\) is large, while our IF-SS-CVT has lower power than the other tests when \(n\) is small. This is expected, as by randomly bisecting each evaluation fold to decouple empirical processes, the IF-SS-CVT operates on an effective sample size of \(n/2\). Despite this inherent finite-sample penalty, the IF-SS-CVT remains remarkably powerful when the sample size is large.

\begin{table}[H]\centering
  \def\sym#1{\ifmmode^{#1}\else\(^{#1}\)\fi}
  \caption{Simulated Rejection Rate}
  \begin{tabular}{lcccccc}
    \hline
                                           &       &       &       &           & \multicolumn{2}{c}{IF-SS-CVT}           \\
                                           & \(n\) & CHIM  & SC    & Naive DML & Lasso                         & XGBoost \\
    \hline
                                           & 250   & 98.8  & 20.2  & 0.7       & 4.0                           & 3.8     \\
    homo (size)                            & 500   & 67.3  & 36.3  & 0.7       & 4.2                           & 4.9     \\
                                           & 1000  & 30.6  & 44.1  & 0.6       & 4.5                           & 5.6     \\
                                           & 2000  & 17.6  & 13.9  & 0.6       & 4.8                           & 5.0     \\
    \hline
                                           & 250   & 100.0 & 100.0 & 100.0     & 95.2                          & 42.1    \\
    linear                                 & 500   & 100.0 & 100.0 & 100.0     & 99.9                          & 89.7    \\
                                           & 1000  & 100.0 & 100.0 & 100.0     & 100.0                         & 100.0   \\
                                           & 2000  & 100.0 & 100.0 & 100.0     & 100.0                         & 100.0   \\
    \hline
                                           & 250   & 100.0 & 100.0 & 100.0     & 56.2                          & 50.4    \\
    kinked                                 & 500   & 100.0 & 100.0 & 100.0     & 87.9                          & 94.6    \\
                                           & 1000  & 100.0 & 100.0 & 100.0     & 99.5                          & 100.0   \\
                                           & 2000  & 100.0 & 100.0 & 100.0     & 100.0                         & 100.0   \\
    \hline
                                           & 250   & 100.0 & 100.0 & 100.0     & 71.3                          & 33.9    \\
    nonlinear                              & 500   & 100.0 & 100.0 & 100.0     & 96.0                          & 77.5    \\
                                           & 1000  & 100.0 & 100.0 & 100.0     & 99.8                          & 97.5    \\
                                           & 2000  & 100.0 & 100.0 & 100.0     & 100.0                         & 100.0   \\
    \hline
  \end{tabular}
  \label{tbl:sim}
  \vspace*{0.09cm}
  \begin{minipage}{0.85\textwidth}
    {\footnotesize
      Empirical rejection proportions, in percentage points, at \(5\%\) significance level based on 1000 simulations. The columns ``CHIM'' and ``SC'' refer to the tests proposed by \citet{crump2008nonparametric} and \citet{semenova2021debiased}, respectively. ``Naive DML'' refers to the standard DML estimator without sample splitting. ``IF-SS-CVT'' refers to the proposed CATE Variance Test in Algorithm \ref{alg:ss_cvt} implemented with Lasso or XGBoost.}
  \end{minipage}
\end{table}

\section{Application \label{sec:app}}

In this section, we demonstrate the application of the proposed test to the NSW job training program data. In this program, participants were randomly assigned to either a job training program or a control group, and the treatment effect on future earnings can be estimated by directly comparing outcomes of the treated and control groups. In order to evaluate the validity of econometric estimators of treatment effects, \citet{lalonde1986evaluating} compared the treated individuals from the experiment to control groups drawn from two survey datasets: the Panel Study of Income Dynamics (PSID) and the Current Population Survey (CPS). The resulting datasets have been extensively analyzed in the influential works by \citet{dehejia1999causal, smith2005does, angrist2009mostly, sloczynski2022interpreting}, among others. In the context of CATE hypothesis testing, the dataset was analyzed by \citet{hsu2017consistent} and \citet{dai2023nonparametric}, who focused specifically on heterogeneity with respect to age. Using the proposed test, we examine heterogeneity with respect to all available covariates.

The dataset we use is NSW-CPS, which contains 185 treated units from the experiment and 15992 control units from the CPS. The outcome \(Y_i\) is the earnings in 1978, and the treatment \(D_i\) is a binary indicator of whether the individual received the job training. We consider the same set of covariates as those in column 4 of Table 3.3.3 in \citet{angrist2009mostly}, which includes age, age squared, education, dummy variables for black and Hispanic, marital status, a dummy indicator for high school degree, and pre-treatment earnings in 1974 and 1975. For this set of covariates \(X_i\), we test \(H_0: \tau(x) = c\) for some constant \(c\) and all covariate values \(x\). For nuisance parameter estimation in the IF-SS-CVT, we employ XGBoost. We compare the results with the CHIM and SC tests introduced in the simulation section, maintaining the same specifications for the basis functions (linear basis for CHIM and second-order polynomials with interactions for SC).

The test results are presented in Table~\ref{tbl:app1}. The CHIM test fails to reject the null hypothesis of constant treatment effects at the 5\% significance level (\(p = 0.40\)). This lack of rejection might be attributed to the test's lower power in finite samples with moderate-dimensional covariates, as observed in our simulations. In contrast, both the SC test and our proposed IF-SS-CVT with Lasso or XGBoost strongly reject the null hypothesis (\(p < 0.01\)), providing robust evidence for the presence of heterogeneous treatment effects. The rejection by the SC test suggests that some of the heterogeneity is linear in the covariates, while the consistent rejection by both Lasso- and XGBoost-based IF-SS-CVT confirms that this finding is not an artifact of a specific machine learning method. These findings complement the conventional ATE-focused analyses by highlighting that the treatment effect of job training likely varies across individuals with different characteristics.

\begin{table}[H]\centering
    \caption{Test Results for NSW Data}
    \label{tbl:app1}
    \begin{tabular}{lcc}
        \hline
        Test                & Statistic    & P-value    \\
        \hline
        CHIM                & Wald = 9.42  & 0.40       \\
        SC                  & Wald = 35.81 & \(<\) 0.01 \\
        IF-SS-CVT (Lasso)   & Z = 9.45     & \(<\) 0.01 \\
        IF-SS-CVT (XGBoost) & Z = 10.78    & \(<\) 0.01 \\
        \hline
    \end{tabular}
\end{table}

\section{Conclusion \label{sec:conclude}}

This paper develops a hypothesis test for the presence of heterogeneous treatment effects by targeting a single omnibus parameter, the variance of the CATE, \(\theta_0 = \Var(\tau_0(X))\). In developing this test, we identify a fundamental theoretical impasse in semiparametric inference at the boundary of the parameter space. On one hand, evaluating variance components on the identical sample leads to null degeneracy, where the asymptotic variance collapses to zero and invalidates standard Gaussian approximations. On the other hand, decoupling the empirical processes via standard sample-splitting destroys the Neyman orthogonality of the squared pseudo-outcomes.

To resolve this impasse, we develop a novel Intra-Fold Sample-Split algorithm. By randomly bisecting each evaluation fold and computing the total and residual variance components on mutually disjoint halves, our procedure guarantees a positive asymptotic variance under the null. By strictly coupling both evaluation halves to identically trained nuisance estimators, the non-orthogonal squared biases cancel out. We formally prove that this algorithm restores Neyman orthogonality, yields asymptotic normality, and guarantees valid Type I error control under the null hypothesis.

Monte Carlo simulations and an empirical application to the NSW job training program confirm the robust finite-sample performance of the proposed test. Our simulations provide empirical proof that, on the boundary, standard cross-fitted DML statistics degenerate and become severely conservative while projection-based HTE tests severely over-reject, whereas our IF-SS-CVT attains the nominal size. Beyond testing for treatment effect moderation, our algorithm provides a general framework for conducting robust hypothesis testing on nonlinear transformations of doubly robust scores.

\clearpage
\begin{singlespace}
  \bibliographystyle{ecta}
  \bibliography{references.bib}
\end{singlespace}

\newpage
\appendix
\setcounter{table}{0}
\renewcommand{\tablename}{Appendix Table}
\renewcommand{\figurename}{Appendix Figure}
\renewcommand{\thetable}{A\arabic{table}}
\setcounter{figure}{0}
\renewcommand{\thefigure}{A\arabic{figure}}

\section{Appendix: Proofs \label{sec:appproofs}}

\setcounter{lemma}{0}
\renewcommand{\thelemma}{A\arabic{lemma}}

\numberwithin{equation}{section}
\setcounter{equation}{0}
\subsection{Proof of Proposition \ref{prp:cvtif}}
\begin{proof}
  We regard the CATE variance \(\theta_0\) as a statistical functional \(\theta(P)\) defined on a nonparametric space of probability distributions \(\mathcal{P}\) satisfying Assumptions \ref{ass:unconfound}--\ref{ass:nuisance_regularity}. Following the identification result in Equation \eqref{eq:cate_var} of the main text, we can express the functional as the difference between two variance components
  \[
    \theta(P) = V_{\mathrm{tot}}(P) - V_{\mathrm{res}}(P),
  \]
  where
  \begin{align*}
    V_{\mathrm{tot}}(P) & = \mathbb{E}_P\left[ (\psi_P(O) - \tau_{\mathrm{ATE}}(P))^2 \right], \\
    V_{\mathrm{res}}(P) & = \mathbb{E}_P\left[ (\psi_P(O) - \tau_P(X))^2 \right].
  \end{align*}
  Here, \(\psi_P(O)\) is the doubly robust pseudo-outcome evaluated using the nuisance parameters \(\eta_P = (e_P, \mu_P)\) under distribution \(P\). The projection parameters are defined as \(\tau_{\mathrm{ATE}}(P) = \mathbb{E}_P[\psi_P(O)]\) and \(\tau_P(X) = \mathbb{E}_P[\psi_P(O) \mid X]\).
  
  To derive the influence function, we compute the pathwise (G\^{a}teaux) derivative of the functional \(\theta(P)\) along a smooth, one-dimensional parametric submodel \(\{P_t : t \in [0, \epsilon)\} \subset \mathcal{P}\) that passes through the true distribution \(P_0\) at \(t=0\). Let \(S(O) = \frac{\partial}{\partial t} \log dP_t(O) \big|_{t=0}\) denote the score function of this submodel. By definition, the influence function \(\phi_\theta(O)\) is the unique mean-zero function satisfying
  \[
    \frac{d}{dt} \theta(P_t) \bigg|_{t=0} = \mathbb{E}_{P_0} \left[ \phi_\theta(O) S(O) \right].
  \]

  Applying the chain rule, the variation of the total variance functional comes from the change in the measure \(P_t\), the nuisance parameter variation in \(\psi_{P_t}(O)\), and the change in the centering parameter \(\tau_{\mathrm{ATE}}(P_t)\)
  \begin{align*}
    \frac{d}{dt} V_{\mathrm{tot}}(P_t) \bigg|_{t=0} & = \mathbb{E}_{P_0} \left[ (\psi_{P_0}(O) - \tau_{\mathrm{ATE}}(P_0))^2 S(O) \right]                                                                                                                                            \\
                                                    & \quad + 2 \mathbb{E}_{P_0} \left[ (\psi_{P_0}(O) - \tau_{\mathrm{ATE}}(P_0)) \left( \frac{\partial \psi_{P_t}(O)}{\partial t}\bigg|_{t=0} - \frac{\partial \tau_{\mathrm{ATE}}(P_t)}{\partial t}\bigg|_{t=0} \right) \right].
  \end{align*}
  Since \(\tau_{\mathrm{ATE}}(P_0) = \mathbb{E}_{P_0}[\psi_{P_0}(O)]\), the expected pseudo-outcome residual strictly vanishes: \(\mathbb{E}_{P_0}[\psi_{P_0}(O) - \tau_{\mathrm{ATE}}(P_0)] = 0\). As a result, the derivative with respect to \(\tau_{\mathrm{ATE}}(P_t)\) evaluates to zero
  \[
    -2 \frac{\partial \tau_{\mathrm{ATE}}(P_t)}{\partial t}\bigg|_{t=0} \mathbb{E}_{P_0} \left[ \psi_{P_0}(O) - \tau_{\mathrm{ATE}}(P_0) \right] = 0.
  \]
  Hence, the pathwise derivative simplifies to
  \begin{equation}
    \label{eq:dt_vtot}
    \frac{d}{dt} V_{\mathrm{tot}}(P_t) \bigg|_{t=0} = \mathbb{E}_{P_0} \left[ (\psi_{P_0}(O) - \tau_{\mathrm{ATE}}(P_0))^2 S(O) \right] + \mathbb{E}_{P_0} \left[ 2 (\psi_{P_0}(O) - \tau_{\mathrm{ATE}}(P_0)) \frac{\partial \psi_{P_t}(O)}{\partial t}\bigg|_{t=0} \right].
  \end{equation}

  Similarly, we differentiate the residual variance functional. The function \(\tau_{P_t}(X)\) minimizes the mean squared error \(\mathbb{E}_{P_t}[(\psi_{P_t}(O) - f(X))^2]\). By the Envelope Theorem for functional optimization (or simply the orthogonal projection property), the first-order variation with respect to the optimal conditional mean function \(\tau_{P_t}(X)\) vanishes. Using the Law of Iterated Expectations, we have
  \[
    \mathbb{E}_{P_0} \left[ 2(\psi_{P_0}(O) - \tau_{P_0}(X)) \frac{\partial \tau_{P_t}(X)}{\partial t}\bigg|_{t=0} \right] = \mathbb{E}_{P_0} \left[ 2 \frac{\partial \tau_{P_t}(X)}{\partial t}\bigg|_{t=0} \underbrace{\mathbb{E}_{P_0} \left[ \psi_{P_0}(O) - \tau_{P_0}(X) \mid X \right]}_{= 0} \right] = 0.
  \]
  Thus, the pathwise derivative of \(V_{\mathrm{res}}(P_t)\) is
  \begin{equation}
    \label{eq:dt_vres}
    \frac{d}{dt} V_{\mathrm{res}}(P_t) \bigg|_{t=0} = \mathbb{E}_{P_0} \left[ (\psi_{P_0}(O) - \tau_{P_0}(X))^2 S(O) \right] + \mathbb{E}_{P_0} \left[ 2 (\psi_{P_0}(O) - \tau_{P_0}(X)) \frac{\partial \psi_{P_t}(O)}{\partial t}\bigg|_{t=0} \right].
  \end{equation}

  As demonstrated in Section \ref{sec:impasse} of the main text, the conditional covariance between the pseudo-outcome and its nuisance derivative, \(\mathbb{E}_{P_0} [ (\psi_{P_0}(O) - \tau_{P_0}(X)) \frac{\partial \psi_{P_t}(O)}{\partial t}|_{t=0} \mid X ]\), is generally non-zero. Thus, \(V_{\mathrm{tot}}(P)\) and \(V_{\mathrm{res}}(P)\) are not individually Neyman orthogonal.

  However, by subtracting \eqref{eq:dt_vres} from \eqref{eq:dt_vtot}, we obtain the pathwise derivative of the target parameter \(\theta(P_t)\). The combined nuisance variation of the difference is
  \begin{align*}
    \Delta_{\mathrm{nuisance}} & = \mathbb{E}_{P_0} \left[ 2 (\psi_{P_0}(O) - \tau_{\mathrm{ATE}}(P_0)) \frac{\partial \psi_{P_t}(O)}{\partial t}\bigg|_{t=0} \right] - \mathbb{E}_{P_0} \left[ 2 (\psi_{P_0}(O) - \tau_{P_0}(X)) \frac{\partial \psi_{P_t}(O)}{\partial t}\bigg|_{t=0} \right] \\
                               & = \mathbb{E}_{P_0} \left[ 2 (\tau_{P_0}(X) - \tau_{\mathrm{ATE}}(P_0)) \frac{\partial \psi_{P_t}(O)}{\partial t}\bigg|_{t=0} \right].
  \end{align*}
  Crucially, the scalar factor \((\tau_{P_0}(X) - \tau_{\mathrm{ATE}}(P_0))\) is purely a function of the covariates \(X\). Applying the Law of Iterated Expectations, we condition on \(X\) first
  \begin{equation}
    \label{eq:dt_nuisance}
    \Delta_{\mathrm{nuisance}} = \mathbb{E}_{P_0} \left[ 2 (\tau_{P_0}(X) - \tau_{\mathrm{ATE}}(P_0)) \ \mathbb{E}_{P_0} \left[ \frac{\partial \psi_{P_t}(O)}{\partial t}\bigg|_{t=0} \ \bigg| \ X \right] \right].
  \end{equation}
  To evaluate the inner conditional expectation, we expand the derivative of the doubly robust pseudo-outcome with respect to the submodel parameter \(t\)
  \begin{align*}
    \frac{\partial \psi_{P_t}(O)}{\partial t}\bigg|_{t=0} & = \frac{\partial \mu_{P_t}(1, X)}{\partial t}\bigg|_{t=0} \left( 1 - \frac{D}{e_{P_0}(X)} \right) - \frac{\partial \mu_{P_t}(0, X)}{\partial t}\bigg|_{t=0} \left( 1 - \frac{1 - D}{1 - e_{P_0}(X)} \right)           \\
                                                          & \quad - \frac{D(Y - \mu_{P_0}(1, X))}{e_{P_0}(X)^2} \frac{\partial e_{P_t}(X)}{\partial t}\bigg|_{t=0} + \frac{(1 - D)(Y - \mu_{P_0}(0, X))}{(1 - e_{P_0}(X))^2} \frac{\partial e_{P_t}(X)}{\partial t}\bigg|_{t=0}.
  \end{align*}
  Taking the conditional expectation given \(X\) under the true probability measure \(P_0\), we use the unconfoundedness assumption \(\mathbb{E}_{P_0}[D \mid X] = e_{P_0}(X)\) and the consistency of the outcome regressions \(\mathbb{E}_{P_0}[Y \mid D=d, X] = \mu_{P_0}(d, X)\)
  \begin{align*}
     & \mathbb{E}_{P_0}\left[1 - \frac{D}{e_{P_0}(X)} \ \bigg| \ X \right] = 1 - \frac{e_{P_0}(X)}{e_{P_0}(X)} = 0                                                                                                            \\
     & \mathbb{E}_{P_0}\left[1 - \frac{1 - D}{1 - e_{P_0}(X)} \ \bigg| \ X \right] = 1 - \frac{1 - e_{P_0}(X)}{1 - e_{P_0}(X)} = 0                                                                                            \\
     & \mathbb{E}_{P_0}\left[\frac{D(Y - \mu_{P_0}(1, X))}{e_{P_0}(X)^2} \ \bigg| \ X \right] = \frac{e_{P_0}(X)}{e_{P_0}(X)^2} \underbrace{\mathbb{E}_{P_0}[Y - \mu_{P_0}(1, X) \mid D=1, X]}_{=0} = 0                       \\
     & \mathbb{E}_{P_0}\left[\frac{(1 - D)(Y - \mu_{P_0}(0, X))}{(1 - e_{P_0}(X))^2} \ \bigg| \ X \right] = \frac{1 - e_{P_0}(X)}{(1 - e_{P_0}(X))^2} \underbrace{\mathbb{E}_{P_0}[Y - \mu_{P_0}(0, X) \mid D=0, X]}_{=0} = 0 \\
  \end{align*}
  Consequently, \(\mathbb{E}_{P_0} \left[ \frac{\partial \psi_{P_t}(O)}{\partial t}\big|_{t=0} \ \big| \ X \right] = 0\). This fundamental property confirms that the expected first-order nuisance bias of the doubly robust pseudo-outcome conditional on \(X\) is exactly zero. Substituting this back into Equation \eqref{eq:dt_nuisance}, the entire joint nuisance variation algebraically annihilates: \(\Delta_{\mathrm{nuisance}} = 0\).

  With the nuisance variation strictly vanishing, the pathwise derivative of the target functional is exclusively driven by the measure variation
  \[
    \frac{d}{dt}\theta(P_t) \bigg|_{t=0} = \mathbb{E}_{P_0} \left[ \left( (\psi_{P_0}(O) - \tau_{\mathrm{ATE}}(P_0))^2 - (\psi_{P_0}(O) - \tau_{P_0}(X))^2 \right) S(O) \right].
  \]
  To officially identify the influence function \(\phi_\theta(O)\) as the Riesz representer, the bracketed term must have an expected value of zero under \(P_0\). Because the score function is mean-zero (\(\mathbb{E}_{P_0}[S(O)] = 0\)), we can subtract the constant \(\theta_0 = \mathbb{E}_{P_0} [ (\psi_{P_0}(O) - \tau_{\mathrm{ATE}}(P_0))^2 - (\psi_{P_0}(O) - \tau_{P_0}(X))^2 ]\) from the integrand without altering the expectation:
  \[
    \frac{d}{dt}\theta(P_t) \bigg|_{t=0} = \mathbb{E}_{P_0} \left[ \Big( (\psi_{P_0}(O) - \tau_{\mathrm{ATE}}(P_0))^2 - (\psi_{P_0}(O) - \tau_{P_0}(X))^2 - \theta_0 \Big) S(O) \right].
  \]
  Letting \(\psi(O_i) = \psi_{P_0}(O_i)\), \(\tau_{\mathrm{ATE}} = \tau_{\mathrm{ATE}}(P_0)\), and \(\tau_0(X_i) = \tau_{P_0}(X_i)\) for notational simplicity, we extract the required influence function for the CATE variance \(\theta_0\)
  \[
    \phi_\theta(O_i) = (\psi(O_i) - \tau_{\mathrm{ATE}})^2 - (\psi(O_i) - \tau_0(X_i))^2 - \theta_0.
  \]
  This completes the proof.
\end{proof}

\subsection{Proof of Theorem \ref{thm:ss_cvt}}

To establish the asymptotic validity of Algorithm \ref{alg:ss_cvt}, we map the estimator into a sum of independent empirical processes. Throughout the proof, let \(P_0\) denote the true probability measure. Let \(\psi_0(O)\) denote the true uncentered efficient influence function (oracle pseudo-outcome) evaluated with the true nuisance parameters \(\eta_0 = (\mu_0, e_0)\). Let \(\mathbb{P}_{n,k,\mathrm{tot}}\) and \(\mathbb{P}_{n,k,\mathrm{res}}\) denote the empirical probability measures over the mutually disjoint evaluation sub-splits \(\mathcal{I}_{k,\mathrm{tot}}\) and \(\mathcal{I}_{k,\mathrm{res}}\). For notational simplicity and without loss of generality, we assume the \(K\) folds are perfectly balanced such that \(n_k = n/K\), and the sub-splits are of equal size \(m_k = m = n/(2K)\). 

\begin{lemma}
  \label{lem:A1}
  Let \(\hat{\eta}_k = (\hat{e}_k, \hat{\mu}_k, \hat{\tau}_k, \hat{\tau}_{\mathrm{ATE},k})\) be the identically trained nuisance estimators on \(\mathcal{I}_{-k}\). Conditional on \(\mathcal{I}_{-k}\), define the expected drift of the paired variance components evaluated at \(\hat{\eta}_k\) as 
  \[
    B_k = P_0\left[ \big(\hat{\psi}_k(O) - \hat{\tau}_{\mathrm{ATE},k}\big)^2 - \big(\hat{\psi}_k(O) - \hat{\tau}_k(X)\big)^2 \mid \mathcal{I}_{-k} \right] - \theta_0 
  \]
  where \(P_0[\cdot \mid \mathcal{I}_{-k}]\) is the expectation taken over a new observation \(O = (Y,D,X)\) independent of \(\mathcal{I}_{-k}\). Under Assumptions \ref{ass:unconfound}--\ref{ass:nuisance_regularity}, \(B_k = o_P(n^{-1/2})\).
\end{lemma}

\begin{proof}
  We expand the difference of the estimated expected losses using the identity \((c-x)^2 - (c-y)^2 = 2c(y-x) + x^2 - y^2\) with \(c = \hat{\psi}_k\), \(x = \hat{\tau}_{\mathrm{ATE},k}\), and \(y = \hat{\tau}_k\),
  \[ \big(\hat{\psi}_k - \hat{\tau}_{\mathrm{ATE},k}\big)^2 - \big(\hat{\psi}_k - \hat{\tau}_k\big)^2 = 2\hat{\psi}_k(\hat{\tau}_k - \hat{\tau}_{\mathrm{ATE},k}) - \hat{\tau}_k^2 + \hat{\tau}_{\mathrm{ATE},k}^2. \]
  Crucially, because Algorithm \ref{alg:ss_cvt} strictly couples both components to the exact same evaluated \(\hat{\psi}_k\), the non-orthogonal \(\hat{\psi}_k^2\) terms cancel out.

  To isolate the estimation errors from the true functions, define the perturbations: \(\Delta \psi = \hat{\psi}_k - \psi_0\), \(\Delta \tau = \hat{\tau}_k - \tau_0\), and \(\Delta c = \hat{\tau}_{\mathrm{ATE},k} - \tau_{\mathrm{ATE}}\). We expand the exact target parameter \(\theta_0\) using the identical algebraic expansion evaluated at the true functions: \(\theta_0 = P_0[ 2\psi_0(\tau_0 - \tau_{\mathrm{ATE}}) - \tau_0^2 + \tau_{\mathrm{ATE}}^2 ]\).

  Subtracting \(\theta_0\) from the conditional expectation of the estimated losses yields the exact drift \(B_k\)
  \[ 
    B_k = P_0\Big[ 2(\psi_0 + \Delta \psi)(\tau_0 + \Delta \tau - \tau_{\mathrm{ATE}} - \Delta c) - (\tau_0 + \Delta \tau)^2 + (\tau_{\mathrm{ATE}} + \Delta c)^2 \mid \mathcal{I}_{-k} \Big] - \theta_0. \]
  Expanding this expression and grouping terms gives
  \begin{align*}
    B_k = & \underbrace{P_0[2\psi_0 \Delta \tau - 2\tau_0 \Delta \tau \mid \mathcal{I}_{-k}]}_{T_1} + \underbrace{P_0[-2\psi_0 \Delta c + 2\tau_{\mathrm{ATE}} \Delta c \mid \mathcal{I}_{-k}]}_{T_2} +               \\
          & \quad + 2P_0[\Delta \psi(\tau_0 - \tau_{\mathrm{ATE}}) \mid \mathcal{I}_{-k}] + 2P_0[\Delta \psi(\Delta \tau - \Delta c) \mid \mathcal{I}_{-k}] - P_0[\Delta \tau^2 \mid \mathcal{I}_{-k}] + \Delta c^2.
  \end{align*}
  By the Law of Iterated Expectations, \(P_0[\psi_0 \mid X] = \tau_0(X)\) and \(P_0[\psi_0] = \tau_{\mathrm{ATE}}\). Therefore, the linear error terms cancel out: \(T_1 = P_0[2\tau_0\Delta \tau - 2\tau_0\Delta \tau \mid \mathcal{I}_{-k}] = 0\), and \(T_2 = -2\tau_{\mathrm{ATE}}\Delta c + 2\tau_{\mathrm{ATE}}\Delta c = 0\). The conditional drift strictly simplifies to four remainder components:
  \[
    B_k = 2 P_0\big[ \Delta \psi (\tau_0 - \tau_{\mathrm{ATE}}) \mid \mathcal{I}_{-k} \big] + 2 P_0\big[ \Delta \psi (\Delta \tau - \Delta c) \mid \mathcal{I}_{-k} \big] - P_0\big[ \Delta \tau^2 \mid \mathcal{I}_{-k} \big] + \Delta c^2.
  \]
  We bound these four components using the \(L_2\) convergence rates in Assumption \ref{ass:nuisance_regularity}.

  By construction of the pseudo-outcome, the conditional expectation of its estimation error is
  \[
    \mathbb{E}_{P_0}[\Delta \psi \mid X, \mathcal{I}_{-k}] = \frac{\hat{e}_k - e_0}{\hat{e}_k}(\hat{\mu}_k(1) - \mu_0(1)) + \frac{\hat{e}_k - e_0}{1 - \hat{e}_k}(\hat{\mu}_k(0) - \mu_0(0)).
  \]
  Because the true functions \(\tau_0(X)\) and \(\tau_{\mathrm{ATE}}\) are uniformly bounded by \(C\) (Assumption \ref{ass:nuisance_regularity}(ii)), their absolute difference is bounded by \(2C\). Taking the absolute value and applying Cauchy-Schwarz to the product of errors gives
  \[
    2\big|P_0\big[ \Delta \psi (\tau_0 - \tau_{\mathrm{ATE}}) \mid \mathcal{I}_{-k}\big]\big| \le \frac{4C}{\xi} \|\hat{e}_k - e_0\|_{P,2} \big( \|\hat{\mu}_k(1) - \mu_0(1)\|_{P,2} + \|\hat{\mu}_k(0) - \mu_0(0)\|_{P,2} \big),
  \]
  where the overlap constant \(1/\xi\) comes from Assumption \ref{ass:overlap}. By Assumption \ref{ass:nuisance_regularity}(i), this is \(o_P(n^{-1/4}) \times o_P(n^{-1/4}) = o_P(n^{-1/2})\).

  For the second cross-term, the factors \(\Delta \tau\) and \(\Delta c\) are measurable with respect to \(X\) and \(\mathcal{I}_{-k}\), so by the Law of Iterated Expectations we may pass to the conditional expectation \(\mathbb{E}_{P_0}[\Delta \psi \mid X, \mathcal{I}_{-k}]\) before bounding. Because the true and estimated outcome functions are uniformly bounded by \(C\) (Assumption \ref{ass:nuisance_regularity}(ii)), the absolute differences \(|\hat{\mu}_k(d, X) - \mu_0(d, X)|\) are bounded by \(2C\) almost surely, and the \(L_2(P_0)\) norm of the conditional expectation satisfies
  \[
    \big\|\mathbb{E}_{P_0}[\Delta \psi \mid X, \mathcal{I}_{-k}]\big\|_{P,2} \le \frac{2C}{\xi} \|\hat{e}_k - e_0\|_{P,2} = o_P(n^{-1/4}).
  \]
  Applying Cauchy-Schwarz to the outer expectation,
  \[
    2\big|P_0\big[ \Delta \psi (\Delta \tau - \Delta c) \mid \mathcal{I}_{-k}\big]\big| \le 2 \big\|\mathbb{E}_{P_0}[\Delta \psi \mid X, \mathcal{I}_{-k}]\big\|_{P,2} \big( \|\Delta \tau\|_{P,2} + |\Delta c| \big) = o_P(n^{-1/2}).
  \]

  \(P_0[\Delta \tau^2 \mid \mathcal{I}_{-k}] = \|\hat{\tau}_k - \tau_0\|_{P,2}^2 = \big( o_P(n^{-1/4}) \big)^2 = o_P(n^{-1/2})\).

  \(\Delta c^2 = (\hat{\tau}_{\mathrm{ATE},k} - \tau_{\mathrm{ATE}})^2 = \big( O_P(n^{-1/2}) \big)^2 = O_P(n^{-1}) = o_P(n^{-1/2})\).

  Summing these bounds confirms that \(B_k = o_P(n^{-1/2})\). Neyman orthogonality for \(\theta_0\) is successfully restored.
\end{proof}

\begin{lemma}
  \label{lem:A2}
  Define the influence functions for the variance components as \(\phi_{\mathrm{tot}}(O_i) = (\psi_0(O_i) - \tau_{\mathrm{ATE}})^2 - V_{\mathrm{tot}}\) and \(\phi_{\mathrm{res}}(O_i) = (\psi_0(O_i) - \tau_0(X_i))^2 - V_{\mathrm{res}}\). The split-sample estimator \(\hat{\theta}_{\mathrm{split}} = \frac{1}{K}\sum_{k=1}^K \hat{\theta}_k\) satisfies the asymptotic expansion
  \[
    \sqrt{n}(\hat{\theta}_{\mathrm{split}} - \theta_0) = \frac{\sqrt{n}}{K} \sum_{k=1}^K \left( \mathbb{P}_{n,k,\mathrm{tot}}[\phi_{\mathrm{tot}}] - \mathbb{P}_{n,k,\mathrm{res}}[\phi_{\mathrm{res}}] \right) + o_P(1).
  \]
\end{lemma}

\begin{proof}
  For a given fold \(k\), the estimated target is \(\hat{\theta}_k = \hat{V}_{\mathrm{tot},k} - \hat{V}_{\mathrm{res},k}\). We decompose the error into an empirical process evaluated at the estimated nuisance parameters, plus the conditional drift \(B_k\)
  \[
    \hat{\theta}_k - \theta_0 = (\mathbb{P}_{n,k,\mathrm{tot}} - P_0)\big[(\hat{\psi}_k - \hat{\tau}_{\mathrm{ATE},k})^2\big] - (\mathbb{P}_{n,k,\mathrm{res}} - P_0)\big[(\hat{\psi}_k - \hat{\tau}_k)^2\big] + B_k.
  \]
  Because \(B_k = o_P(n^{-1/2})\) as in Lemma \ref{lem:A1} and the out-of-fold estimators \(\hat{\eta}_k\) are computed on the independent set \(\mathcal{I}_{-k}\), we invoke standard stochastic equicontinuity results for cross-fitted Double/Debiased Machine Learning by \citet{chernozhukov2018double}. Crucially, Assumption \ref{ass:nondegen}(i) ensures the outcome \(Y_i\) has bounded fourth moments and Assumption \ref{ass:nuisance_regularity}(ii) ensures the true and estimated nuisance functions are uniformly bounded; together these guarantee that the squared pseudo-outcomes possess a valid square-integrable envelope, satisfying the regularity conditions for cross-fitted empirical processes. Given the \(L_2\) consistency required by Assumption \ref{ass:nuisance_regularity}, substituting the estimated nuisance functions with their true probability limits inside the centered empirical process introduces only an \(o_P(n^{-1/2})\) remainder
  \[
    (\mathbb{P}_{n,k,\mathrm{tot}} - P_0)\big[(\hat{\psi}_k - \hat{\tau}_{\mathrm{ATE},k})^2\big] = (\mathbb{P}_{n,k,\mathrm{tot}} - P_0)\big[(\psi_0 - \tau_{\mathrm{ATE}})^2\big] + o_P(n^{-1/2}) = \mathbb{P}_{n,k,\mathrm{tot}}[\phi_{\mathrm{tot}}] + o_P(n^{-1/2})
  \]
  \[
    (\mathbb{P}_{n,k,\mathrm{res}} - P_0)\big[(\hat{\psi}_k - \hat{\tau}_k)^2\big] = (\mathbb{P}_{n,k,\mathrm{res}} - P_0)\big[(\psi_0 - \tau_0)^2\big] + o_P(n^{-1/2}) = \mathbb{P}_{n,k,\mathrm{res}}[\phi_{\mathrm{res}}] + o_P(n^{-1/2}).
  \]
  Averaging across the \(K\) fixed folds and multiplying by \(\sqrt{n}\) yields the stated linear expansion.
\end{proof}

\begin{lemma}
  \label{lem:A3}
  Let \(\sigma_{\mathrm{tot}}^2 = \Var(\phi_{\mathrm{tot}}(O))\) and \(\sigma_{\mathrm{res}}^2 = \Var(\phi_{\mathrm{res}}(O))\). Under Assumptions~\ref{ass:unconfound}--\ref{ass:nuisance_regularity}, the aggregated test statistic \(\hat{\theta}_{\mathrm{split}}\) satisfies
  \[
    \sqrt{n}(\hat{\theta}_{\mathrm{split}} - \theta_0) \xrightarrow{d} \mathcal{N}\big(0, 2\sigma_{\mathrm{tot}}^2 + 2\sigma_{\mathrm{res}}^2\big).
  \]
  Crucially, the asymptotic variance remains strictly positive even under the null hypothesis \(H_0: \theta_0 = 0\).
\end{lemma}

\begin{proof}
  Let \(\mathcal{I}_{\mathrm{tot}} = \bigcup_{k=1}^K \mathcal{I}_{k,\mathrm{tot}}\) and \(\mathcal{I}_{\mathrm{res}} = \bigcup_{k=1}^K \mathcal{I}_{k,\mathrm{res}}\). By the construction of Algorithm \ref{alg:ss_cvt}, these global sets completely partition the dataset such that \(\mathcal{I}_{\mathrm{tot}} \cap \mathcal{I}_{\mathrm{res}} = \emptyset\), and both sets have an identical size of \(n/2\). We rewrite the summation from Lemma \ref{lem:A2} as a scaled difference over these global sets
  \[
    \sqrt{n}(\hat{\theta}_{\mathrm{split}} - \theta_0) = \sqrt{2} \left[ \frac{1}{\sqrt{n/2}} \sum_{i \in \mathcal{I}_{\mathrm{tot}}} \phi_{\mathrm{tot}}(O_i) \right] - \sqrt{2} \left[ \frac{1}{\sqrt{n/2}} \sum_{j \in \mathcal{I}_{\mathrm{res}}} \phi_{\mathrm{res}}(O_j) \right] + o_P(1).
  \]
  By the Central Limit Theorem, the bracketed terms converge in distribution to \(\mathcal{N}(0, \sigma_{\mathrm{tot}}^2)\) and \(\mathcal{N}(0, \sigma_{\mathrm{res}}^2)\). Because the evaluation subsets \(\mathcal{I}_{\mathrm{tot}}\) and \(\mathcal{I}_{\mathrm{res}}\) are mutually disjoint, the two empirical processes evaluate statistically independent observations. Thus, their covariance is exactly zero, and the variance of their difference sums strictly to \(2\sigma_{\mathrm{tot}}^2 + 2\sigma_{\mathrm{res}}^2\).

  Under \(H_0: \theta_0 = 0\), the CATE is constant almost surely \(\tau_0(X) = \tau_{\mathrm{ATE}}\). Consequently, the true uncentered influence functions become perfectly identical: \(\phi_{\mathrm{tot}}(O) = \phi_{\mathrm{res}}(O)\) almost surely. If these components were evaluated on the same sample, the variance of their difference would degenerate identically to \(\Var(\phi_{\mathrm{tot}} - \phi_{\mathrm{tot}}) = 0\). However, because our algorithm strictly partitions the evaluation sets, the empirical covariance remains exactly zero, and the asymptotic limit variance strictly evaluates to \(4\sigma_{\mathrm{tot}}^2\). By the non-degeneracy condition in Assumption \ref{ass:nondegen}(ii), \(\Var\!\left((\psi_0(O) - \tau_{\mathrm{ATE}})^2\right) = \sigma_{\mathrm{tot}}^2 > c > 0\). Therefore, the asymptotic variance strictly bounds away from zero, formally resolving the null degeneracy. 
\end{proof}

\begin{proof}[Proof of Theorem \ref{thm:ss_cvt}]
  Equipped with the preceding lemmas, we now establish the final asymptotic validity of Algorithm \ref{alg:ss_cvt}.

  The standard error estimator \(\widehat{SE}^2\) proposed in Algorithm \ref{alg:ss_cvt} aggregates the sample variances within the paired sub-splits. Substituting \(m_k = n/(2K)\) and scaling the expression by \(n\)
  \[
    n \widehat{SE}^2 = n \frac{1}{K^2} \sum_{k=1}^K \left( \frac{\hat{\sigma}_{\mathrm{tot},k}^2}{n/(2K)} + \frac{\hat{\sigma}_{\mathrm{res},k}^2}{n/(2K)} \right) = \frac{2}{K} \sum_{k=1}^K \big( \hat{\sigma}_{\mathrm{tot},k}^2 + \hat{\sigma}_{\mathrm{res},k}^2 \big).
  \]
  Under Assumption \ref{ass:nondegen} and the \(L_2\) consistency of the nuisance parameters, the sample variances are weakly consistent for their population counterparts: \(\hat{\sigma}_{\mathrm{tot},k}^2 \xrightarrow{p} \sigma_{\mathrm{tot}}^2\) and \(\hat{\sigma}_{\mathrm{res},k}^2 \xrightarrow{p} \sigma_{\mathrm{res}}^2\). By the Weak Law of Large Numbers over the fixed \(K\) folds, \(n \widehat{SE}^2 \xrightarrow{p} 2\sigma_{\mathrm{tot}}^2 + 2\sigma_{\mathrm{res}}^2\).

  The standardized test statistic can be rewritten as
  \[
    Z_\theta = \frac{\sqrt{n}(\hat{\theta}_{\mathrm{split}} - \theta_0)}{\sqrt{n\widehat{SE}^2}}.
  \]
  By Lemma \ref{lem:A3}, the numerator converges in distribution to \(\mathcal{N}(0, 2\sigma_{\mathrm{tot}}^2 + 2\sigma_{\mathrm{res}}^2)\). The denominator converges in probability to the square root of that identical quantity. Applying Slutsky's Theorem yields the stated standard normal limit
  \[
    Z_\theta \xrightarrow{d} \mathcal{N}(0, 1).
  \]
  
  Under the null hypothesis \(H_0: \theta_0 = 0\), the standardized statistic simplifies to \(Z_\theta = \hat{\theta}_{\mathrm{split}} / \widehat{SE}\). Because the asymptotic Gaussian limit is valid on the boundary (as established by the non-degeneracy in Lemma \ref{lem:A3}), the one-sided test provides exact asymptotic size control
  \[
    \lim_{n \rightarrow \infty} P(Z_{\theta} > z_{1-\alpha} \mid H_0) = 1 - \Phi(z_{1-\alpha}) = \alpha.
  \]
  
  Under the alternative \(H_1: \theta_0 > 0\), the standardized test statistic decomposes into a centered distribution and a non-stochastic drift
  \[
    Z_\theta = \frac{\hat{\theta}_{\mathrm{split}}}{\widehat{SE}} = \frac{\hat{\theta}_{\mathrm{split}} - \theta_0}{\widehat{SE}} + \frac{\theta_0}{\widehat{SE}}.
  \]
  The first term converges to \(\mathcal{N}(0, 1)\) and is \(O_P(1)\). For the drift component, since \(\theta_0 > 0\) is a fixed positive constant and \(\widehat{SE} = O_P(n^{-1/2})\), we have:
  \[
    \frac{\theta_0}{\widehat{SE}} = \sqrt{n} \frac{\theta_0}{\sqrt{n\widehat{SE}^2}} \xrightarrow{p} +\infty.
  \]
  Therefore, the test statistic \(Z_\theta\) diverges to \(+\infty\) at a \(\sqrt{n}\)-rate. Consequently, the probability of rejecting the null hypothesis approaches 1, establishing asymptotic unit power against any fixed alternative
  \[
    \lim_{n \rightarrow \infty} P(Z_{\theta} > z_{1-\alpha} \mid H_1) = 1.
  \]
\end{proof}

\newpage
\section{Appendix: Asymptotic Behavior of the Naive DML Statistic \label{sec:appnaive}}

\setcounter{proposition}{0}
\renewcommand{\theproposition}{B\arabic{proposition}}
\setcounter{assumption}{0}
\renewcommand{\theassumption}{B\arabic{assumption}}
\setcounter{remark}{0}
\renewcommand{\theremark}{B\arabic{remark}}
\setcounter{equation}{0}
\renewcommand{\theHproposition}{B\arabic{proposition}}
\renewcommand{\theHassumption}{B\arabic{assumption}}
\renewcommand{\theHremark}{B\arabic{remark}}

This appendix characterizes the limiting behavior of the Naive DML benchmark introduced in Section \ref{sec:sim}. The benchmark is constructed to isolate the role of the Intra-Fold Sample-Split: it retains standard cross-fitting but evaluates the total and residual variance components on the \emph{same} observations. We show that, under the null hypothesis and an exact-rate condition on the DR-learner, this single change produces a standardized statistic that drifts to \(-\infty\), so the one-sided test's rejection probability converges to zero. This formalizes the conservative degeneracy described in Section \ref{sec:impasse} and accounts for the small, \(n\)-insensitive rejection rates in the Naive DML column of Table \ref{tbl:sim}. The result is proved for the feasible \(K\)-fold statistic by conditional moment arguments alone; no central limit theorem is invoked, so no condition on the dependence between folds induced by the shared training data is required.

For each fold \(k\), the benchmark trains a single set of out-of-fold nuisance estimators \(\hat{\eta}_k = (\hat{e}_k, \hat{\mu}_k)\), a DR-learner \(\hat{\tau}_k\), and an ATE estimator \(\hat{\tau}_{\mathrm{ATE},k}\) on \(\mathcal{I}_{-k}\), and evaluates both variance components on the entire evaluation fold \(\mathcal{I}_k\),
\[
  \hat{V}_{\mathrm{tot},k} = \frac{1}{n_k} \sum_{i \in \mathcal{I}_k} \big(\hat{\psi}_{i,k} - \hat{\tau}_{\mathrm{ATE},k}\big)^2, \qquad \hat{V}_{\mathrm{res},k} = \frac{1}{n_k} \sum_{i \in \mathcal{I}_k} \big(\hat{\psi}_{i,k} - \hat{\tau}_k(X_i)\big)^2.
\]
The point estimate and standard error are
\[
  \hat{\theta}_{\mathrm{naive}} = \frac{1}{K} \sum_{k=1}^K \big(\hat{V}_{\mathrm{tot},k} - \hat{V}_{\mathrm{res},k}\big), \qquad \widehat{SE}_{\mathrm{naive}} = \sqrt{\frac{1}{K^2} \sum_{k=1}^K \frac{\hat{\sigma}^2_{\phi,k}}{n_k}},
\]
the standardized statistic is \(Z_{\mathrm{naive}} = \hat{\theta}_{\mathrm{naive}} / \widehat{SE}_{\mathrm{naive}}\), where \(\hat{\sigma}^2_{\phi,k}\) is the sample variance over \(\mathcal{I}_k\) of the per-unit influence term
\begin{equation}
  \label{eq:naive_phi}
  \hat{\phi}_{i,k} = \big(\hat{\psi}_{i,k} - \hat{\tau}_{\mathrm{ATE},k}\big)^2 - \big(\hat{\psi}_{i,k} - \hat{\tau}_k(X_i)\big)^2,
\end{equation}
and the test rejects \(H_0\) when \(Z_{\mathrm{naive}} > z_{1-\alpha}\). As in the proof of Theorem \ref{thm:ss_cvt} in Appendix \ref{sec:appproofs}, we take the folds balanced, \(n_k = n/K\), without loss of generality.

\paragraph{Notation.} Fix a fold \(k\) and write, with all conditional expectations taken over a new observation \(O = (Y, D, X)\) independent of \(\mathcal{I}_{-k}\),
\begin{equation}
  \label{eq:naive_notation}
  \begin{aligned}
    \hat{h}_k(x)   & := \hat{\tau}_k(x) - \hat{\tau}_{\mathrm{ATE},k},                  & \qquad
    \Delta_{c,k}   & := \hat{\tau}_{\mathrm{ATE},k} - \tau_{\mathrm{ATE}},                        \\
    \delta_{n,k}   & := \|\hat{\tau}_k - \tau_0\|_{P,2},                                & \qquad
    \bar{b}_k      & := P_0\big[\hat{\tau}_k(X) - \tau_0(X) \mid \mathcal{I}_{-k}\big],           \\
    \rho_k         & := P_0\big[\hat{r}_k(X)\, \hat{h}_k(X) \mid \mathcal{I}_{-k}\big], &          &
  \end{aligned}
\end{equation}
where
\begin{equation}
  \label{eq:naive_rhat}
  \hat{r}_k(x) := \big(\hat{e}_k(x) - e_0(x)\big) \left[ \frac{\hat{\mu}_k(1,x) - \mu_0(1,x)}{\hat{e}_k(x)} + \frac{\hat{\mu}_k(0,x) - \mu_0(0,x)}{1 - \hat{e}_k(x)} \right]
\end{equation}
is the doubly robust product bias: as computed in the proof of Lemma \ref{lem:A1}, \(P_0[\hat{\psi}_{i,k} - \psi_0(O_i) \mid X_i, \mathcal{I}_{-k}] = \hat{r}_k(X_i)\). Finally let
\begin{gather*}
  \hat{v}_k(x) := \Var\big(\hat{\psi}_{i,k} \mid X_i = x, \mathcal{I}_{-k}\big), \\
  m_k := P_0\big[\hat{\phi}_{i,k} \mid \mathcal{I}_{-k}\big], \qquad s_k^2 := \Var\big(\hat{\phi}_{i,k} \mid \mathcal{I}_{-k}\big).
\end{gather*}

\begin{assumption}
  \label{ass:naive}
  In addition to Assumptions~\ref{ass:unconfound}--\ref{ass:nuisance_regularity}, the following hold.

  (i) Conditional moments: there exist constants \(\bar{M} < \infty\) and \(\sigma^2_{\min} > 0\) such that, almost surely, \(\mathbb{E}\big[(Y - \mu_0(D,X))^4 \mid D, X\big] \le \bar{M}\) and \(\sigma_d^2(X) = \Var(Y \mid D = d, X) \ge \sigma^2_{\min}\) for \(d \in \{0,1\}\).

  (ii) Exact learner rate: there is a deterministic sequence \(d_n \to 0\) with \(\delta_{n,k}/d_n \xrightarrow{p} 1\) for each \(k = 1, \dots, K\), and \(\sqrt{n}\, d_n \to \infty\). (The unit normalization is without loss of generality; \(\delta_{n,k} \asymp_P d_n\) suffices, with the obvious changes of constants.)

  (iii) Product-bias domination: \(\|\hat{r}_k\|_{P,2} = o_P(d_n)\) for each \(k\). Sufficient primitive conditions: the propensity score is known (then \(\hat{r}_k \equiv 0\), as in a randomized experiment), or \(\|\hat{e}_k - e_0\|_{P,2} = o_P(d_n)\), since \(|\hat{r}_k| \le (4C/\xi)\, |\hat{e}_k - e_0|\) on the event of Assumption~\ref{ass:nuisance_regularity}(ii).
\end{assumption}

Assumption~\ref{ass:naive}(ii) restricts attention to learners whose \(L_2\) error, while vanishing, decays more slowly than the parametric rate; this is the empirically relevant case for regularized machine learners and, as Remark~\ref{rem:boundary} explains, it is also the only regime in which the benchmark admits a one-signed asymptotic verdict under Assumptions~\ref{ass:unconfound}--\ref{ass:nuisance_regularity}. Note that (ii) together with Assumption~\ref{ass:nuisance_regularity}(i) \emph{implies} \(\Delta_{c,k}^2 = O_P(n^{-1}) = o_P(d_n^2)\): the dominance of the learner error over the ATE error is derived, not assumed.

\begin{proposition}[Degeneracy of the Naive DML Test]
  \label{prp:naive}
  Suppose Assumptions~\ref{ass:unconfound}--\ref{ass:nuisance_regularity} and~\ref{ass:naive} hold, and let \(H_0: \theta_0 = 0\) be true. Then:
  \begin{enumerate}
    \item[(i)] \emph{(Exact bias anatomy.)} For every fold \(k\),
          \begin{equation}
            \label{eq:naive_mean}
            m_k = P_0\big[\hat{\phi}_{i,k} \mid \mathcal{I}_{-k}\big] = -\,\delta_{n,k}^2 + \Delta_{c,k}^2 + 2\rho_k,
          \end{equation}
          an identity that uses only Assumptions~\ref{ass:unconfound}--\ref{ass:nondegen} and holds on the event of Assumption~\ref{ass:nuisance_regularity}(ii) (probability approaching one), on which all conditional moments above exist.
    \item[(ii)] \(\hat{\theta}_{\mathrm{naive}} = -\,d_n^2\big(1 + o_P(1)\big)\).
    \item[(iii)] \(n \widehat{SE}^2_{\mathrm{naive}} = S_n^2\big(1 + o_P(1)\big)\), where \(S_n^2 := K^{-1} \sum_{k=1}^K s_k^2\) obeys the two-sided bounds
          \[
            4\sigma^2_{\min} + o_P(1) \le \frac{S_n^2}{d_n^2} \le C^{\star} + o_P(1)
          \]
          for a finite constant \(C^{\star}\) depending only on \((\xi, C, \bar{M})\). In particular the naive standard error is of exact order \(d_n/\sqrt{n}\).
    \item[(iv)] Consequently
          \begin{gather*}
            Z_{\mathrm{naive}} = -\,\frac{\sqrt{n}\, d_n^2}{S_n}\big(1 + o_P(1)\big) \xrightarrow{p} -\infty, \\
            \lim_{n \to \infty} P\big(Z_{\mathrm{naive}} > z_{1-\alpha} \mid H_0\big) = 0 \quad \text{for every } \alpha \in (0,1).
          \end{gather*}
          The benchmark therefore never over-rejects asymptotically; its asymptotic size is zero, and the standardized statistic diverges at the rate \(\sqrt{n}\, d_n\).
  \end{enumerate}
\end{proposition}

\begin{proof}
  Fix a fold \(k\); all constants \(C_1, C_2, \dots\) below depend only on \((\xi, C, \bar{M})\), and \(K\) is fixed. Since \(|\tau_0| \le C\) implies \(|\tau_{\mathrm{ATE}}| \le C\), the event
  \[
    \mathcal{E}_{n,k} := \Big\{ \hat{e}_k(X) \in [\xi, 1-\xi], \; \max\big(|\hat{\mu}_k(d,X)|, |\hat{\tau}_k(X)|\big) \le C \ \text{a.s.}, \; |\Delta_{c,k}| \le 1 \Big\}
  \]
  satisfies \(P(\mathcal{E}_{n,k}) \to 1\): the first two requirements hold with probability approaching one by Assumption~\ref{ass:nuisance_regularity}(ii), and the third holds because Assumption~\ref{ass:nuisance_regularity}(i) gives \(\Delta_{c,k} = O_P(n^{-1/2}) = o_P(1)\). Since every conclusion is a statement of convergence in probability, we may and do argue on \(\mathcal{E}_{n,k}\) throughout. On this event \(|\hat{\tau}_{\mathrm{ATE},k}| \le C + 1\), hence \(|\hat{h}_k| \le \bar{h} := 2C + 1\) and, by \eqref{eq:naive_rhat}, \(|\hat{r}_k| \le \bar{r} := 4C/\xi\) almost surely.

  \emph{Step 1 (conditional structure).} Let \(\varepsilon_{i,k} := \hat{\psi}_{i,k} - P_0[\hat{\psi}_{i,k} \mid X_i, \mathcal{I}_{-k}]\), so that \(P_0[\varepsilon_{i,k} \mid X_i, \mathcal{I}_{-k}] = 0\) and \(P_0[\varepsilon_{i,k}^2 \mid X_i, \mathcal{I}_{-k}] = \hat{v}_k(X_i)\). Since \(P_0[\psi_0 \mid X] = \tau_0(X) = \tau_{\mathrm{ATE}}\) under \(H_0\) and \(P_0[\hat{\psi}_{i,k} - \psi_0 \mid X_i, \mathcal{I}_{-k}] = \hat{r}_k(X_i)\),
  \[
    \hat{\psi}_{i,k} - \hat{\tau}_{\mathrm{ATE},k} = \varepsilon_{i,k} + \hat{r}_k(X_i) - \Delta_{c,k}.
  \]
  Since \(\hat{\psi}_{i,k} - \hat{\tau}_k(X_i) = (\hat{\psi}_{i,k} - \hat{\tau}_{\mathrm{ATE},k}) - \hat{h}_k(X_i)\), the influence term \eqref{eq:naive_phi} factorizes as
  \begin{equation}
    \label{eq:naive_factor}
    \hat{\phi}_{i,k} = 2\big(\hat{\psi}_{i,k} - \hat{\tau}_{\mathrm{ATE},k}\big) \hat{h}_k(X_i) - \hat{h}_k(X_i)^2 = 2\, \varepsilon_{i,k}\, \hat{h}_k(X_i) + q_k(X_i),
  \end{equation}
  with \(q_k(x) := 2\big(\hat{r}_k(x) - \Delta_{c,k}\big) \hat{h}_k(x) - \hat{h}_k(x)^2\), an \(\mathcal{I}_{-k}\)-measurable function of \(x\) alone.

  \emph{Step 2 (exact mean; part (i)).} Because
  \[
    P_0\big[\varepsilon_{i,k} \hat{h}_k(X_i) \mid \mathcal{I}_{-k}\big] = P_0\big[\hat{h}_k(X)\, P_0[\varepsilon_{i,k} \mid X, \mathcal{I}_{-k}]\big] = 0,
  \]
  we have \(m_k = P_0[q_k \mid \mathcal{I}_{-k}]\). Under \(H_0\), \(\hat{h}_k = (\hat{\tau}_k - \tau_0) - \Delta_{c,k}\), so
  \[
    P_0[\hat{h}_k \mid \mathcal{I}_{-k}] = \bar{b}_k - \Delta_{c,k}, \qquad P_0[\hat{h}_k^2 \mid \mathcal{I}_{-k}] = \delta_{n,k}^2 - 2\Delta_{c,k}\bar{b}_k + \Delta_{c,k}^2.
  \]
  Therefore
  \[
    m_k = 2\rho_k - 2\Delta_{c,k}\big(\bar{b}_k - \Delta_{c,k}\big) - \big(\delta_{n,k}^2 - 2\Delta_{c,k}\bar{b}_k + \Delta_{c,k}^2\big) = -\,\delta_{n,k}^2 + \Delta_{c,k}^2 + 2\rho_k,
  \]
  which is \eqref{eq:naive_mean}; the terms in \(\Delta_{c,k}\bar{b}_k\) cancel exactly. Only integrability and the conditional-bias formula were used, so (i) holds under Assumptions~\ref{ass:unconfound}--\ref{ass:nondegen} alone on the event of Assumption~\ref{ass:nuisance_regularity}(ii).

  \emph{Step 3 (drift).} By Jensen's inequality \(|\bar{b}_k| \le \delta_{n,k}\), and by Assumption~\ref{ass:nuisance_regularity}(i) and \(\sqrt{n}\, d_n \to \infty\) we have \(\Delta_{c,k} = O_P(n^{-1/2}) = o_P(d_n)\). Hence, using Assumption~\ref{ass:naive}(ii),
  \begin{equation}
    \label{eq:naive_hnorm}
    \|\hat{h}_k\|_{P,2}^2 = \delta_{n,k}^2 - 2\Delta_{c,k}\bar{b}_k + \Delta_{c,k}^2 = d_n^2\big(1 + o_P(1)\big).
  \end{equation}
  By Cauchy--Schwarz and Assumption~\ref{ass:naive}(iii), \(|\rho_k| \le \|\hat{r}_k\|_{P,2} \|\hat{h}_k\|_{P,2} = o_P(d_n) \cdot O_P(d_n) = o_P(d_n^2)\), while \(\Delta_{c,k}^2 = O_P(n^{-1}) = o_P(d_n^2)\). Combining with (i),
  \begin{equation}
    \label{eq:naive_mk}
    m_k = -\,d_n^2\big(1 + o_P(1)\big).
  \end{equation}

  \emph{Step 4 (two-sided conditional variance bounds).} Since \(P_0[\varepsilon_{i,k} \hat{h}_k (q_k - m_k) \mid \mathcal{I}_{-k}] = P_0[\hat{h}_k (q_k - m_k)\, P_0[\varepsilon_{i,k} \mid X, \mathcal{I}_{-k}]] = 0\), the decomposition \eqref{eq:naive_factor} gives
  \begin{equation}
    \label{eq:naive_svar}
    s_k^2 = 4\, P_0\big[\hat{v}_k(X) \hat{h}_k(X)^2 \mid \mathcal{I}_{-k}\big] + \Var\big(q_k(X) \mid \mathcal{I}_{-k}\big).
  \end{equation}
  For the lower bound, the law of total variance and \(\hat{e}_k, 1 - \hat{e}_k \le 1\) give, on \(\mathcal{E}_{n,k}\),
  \begin{align*}
    \hat{v}_k(X) & \ge P_0\big[\Var(\hat{\psi}_{i,k} \mid D, X) \mid X\big]                                                                                            \\
                 & = e_0(X)\, \frac{\sigma_1^2(X)}{\hat{e}_k(X)^2} + \big(1 - e_0(X)\big)\, \frac{\sigma_0^2(X)}{(1 - \hat{e}_k(X))^2} \;\ge\; \sigma^2_{\min},
  \end{align*}
  so \(s_k^2 \ge 4\sigma^2_{\min} \|\hat{h}_k\|_{P,2}^2\). For the upper bound, \(|\hat{\psi}_{i,k}| \le 2C + (|Y - \mu_0(D,X)| + 2C)/\xi\) on \(\mathcal{E}_{n,k}\) and \(\mathbb{E}[(Y - \mu_0(D,X))^2 \mid D, X] \le \bar{M}^{1/2}\) yield \(\hat{v}_k(X) \le P_0[\hat{\psi}_{i,k}^2 \mid X, \mathcal{I}_{-k}] \le C_1\); moreover \(|q_k| \le \big(2(\bar{r} + 1) + \bar{h}\big) |\hat{h}_k|\) on \(\mathcal{E}_{n,k}\) (using \(|\Delta_{c,k}| \le 1\) there), so \(\Var(q_k \mid \mathcal{I}_{-k}) \le P_0[q_k^2 \mid \mathcal{I}_{-k}] \le C_2 \|\hat{h}_k\|_{P,2}^2\). Hence, with \(C_3 := 4C_1 + C_2\),
  \begin{equation}
    \label{eq:naive_svar_bounds}
    4\sigma^2_{\min}\, \|\hat{h}_k\|_{P,2}^2 \;\le\; s_k^2 \;\le\; C_3\, \|\hat{h}_k\|_{P,2}^2,
  \end{equation}
  and by \eqref{eq:naive_hnorm} both bounds are of exact order \(d_n^2\).

  \emph{Step 5 (fluctuations; part (ii)).} Let \(A_k := n_k^{-1} \sum_{i \in \mathcal{I}_k} (\hat{\phi}_{i,k} - m_k)\), so that \(P_0[A_k \mid \mathcal{I}_{-k}] = 0\) and \(\Var(A_k \mid \mathcal{I}_{-k}) = s_k^2 / n_k\), the observations in \(\mathcal{I}_k\) being i.i.d.\ and independent of \(\mathcal{I}_{-k}\). By conditional Chebyshev, \eqref{eq:naive_svar_bounds}, \eqref{eq:naive_hnorm} and \(P(\mathcal{E}_{n,k}) \to 1\), for every \(\eta > 0\)
  \[
    P\big(|A_k| > \eta\, d_n^2\big) \le \mathbb{E}\left[ \min\left\{ 1, \frac{s_k^2}{n_k\, \eta^2 d_n^4} \right\} \right] + o(1) \le \frac{2 C_3 K}{\eta^2\, n\, d_n^2} + o(1) \longrightarrow 0,
  \]
  using \(n\, d_n^2 \to \infty\). Thus \(A_k = o_P(d_n^2)\) for each of the \(K\) folds separately --- no joint moment across folds is required --- and with \eqref{eq:naive_mk},
  \[
    \hat{\theta}_{\mathrm{naive}} = \frac{1}{K} \sum_{k=1}^K \big(m_k + A_k\big) = -\,d_n^2\big(1 + o_P(1)\big),
  \]
  which is (ii). Note this already implies \(P(\hat{\theta}_{\mathrm{naive}} \ge 0) \to 0\) and hence, since \(\widehat{SE}_{\mathrm{naive}} > 0\) and \(z_{1-\alpha} > 0\), the size conclusion in (iv) --- the remaining steps sharpen this to the stated rate.

  \emph{Step 6 (standard error; part (iii)).} Write \(\hat{\sigma}_{\phi,k}^2 = n_k^{-1} \sum_{i \in \mathcal{I}_k} \hat{\phi}_{i,k}^2 - (m_k + A_k)^2\) (the \(n_k/(n_k - 1)\) correction is immaterial). On \(\mathcal{E}_{n,k}\), \(\mathbb{E}[\varepsilon_{i,k}^4 \mid X_i, \mathcal{I}_{-k}] \le C_4(1 + \bar{M})\) by the same envelope as in Step~4, so from \eqref{eq:naive_factor}, \(|\hat{h}_k| \le \bar{h}\) and \(|q_k| \le C_5\),
  \[
    P_0\big[\hat{\phi}_{i,k}^4 \mid \mathcal{I}_{-k}\big] \le 8\Big( 16\, \bar{h}^2\, C_4(1 + \bar{M})\, \|\hat{h}_k\|_{P,2}^2 + C_5^2\, P_0[q_k^2 \mid \mathcal{I}_{-k}] \Big) \le C_6\, \|\hat{h}_k\|_{P,2}^2.
  \]
  Since \(P_0[\hat{\phi}_{i,k}^2 \mid \mathcal{I}_{-k}] = s_k^2 + m_k^2\), conditional Chebyshev gives
  \[
    n_k^{-1} \sum_{i \in \mathcal{I}_k} \hat{\phi}_{i,k}^2 = s_k^2 + m_k^2 + O_P\Big( \sqrt{C_6\, d_n^2 / n_k} \Big) = s_k^2 + o_P(d_n^2),
  \]
  because \(m_k^2 = O_P(d_n^4) = o_P(d_n^2)\) and \(d_n/\sqrt{n} = o_P(d_n^2)\) by \(\sqrt{n}\, d_n \to \infty\). Likewise \((m_k + A_k)^2 = O_P(d_n^4) = o_P(d_n^2)\). Hence \(\hat{\sigma}_{\phi,k}^2 = s_k^2 + o_P(d_n^2) = s_k^2\big(1 + o_P(1)\big)\), the last step by the lower bound in \eqref{eq:naive_svar_bounds}. With \(n_k = n/K\),
  \[
    n \widehat{SE}^2_{\mathrm{naive}} = \frac{1}{K} \sum_{k=1}^K \hat{\sigma}_{\phi,k}^2 = S_n^2\big(1 + o_P(1)\big),
  \]
  and the stated bounds on \(S_n^2 / d_n^2\) follow from \eqref{eq:naive_svar_bounds} and \eqref{eq:naive_hnorm} with \(C^{\star} := C_3\).

  \emph{Step 7 (conclusion; part (iv)).} Combining (ii) and (iii),
  \[
    Z_{\mathrm{naive}} = \frac{\sqrt{n}\, \hat{\theta}_{\mathrm{naive}}}{\sqrt{n \widehat{SE}^2_{\mathrm{naive}}}} = -\,\frac{\sqrt{n}\, d_n^2}{S_n}\big(1 + o_P(1)\big),
  \]
  and hence, by the upper bound in (iii),
  \[
    -\,Z_{\mathrm{naive}} \;\ge\; \frac{\sqrt{n}\, d_n}{\sqrt{C^{\star}}}\, \big(1 + o_P(1)\big) \xrightarrow{p} \infty,
  \]
  so \(Z_{\mathrm{naive}} \xrightarrow{p} -\infty\) and, since \(z_{1-\alpha}\) is fixed, \(P(Z_{\mathrm{naive}} > z_{1-\alpha}) \to 0\), proving (iv).
\end{proof}

\begin{remark}[The boundary regime \(\sqrt{n}\, d_n = O(1)\), and why no Gaussian limit is claimed]
  \label{rem:boundary}
  Assumption~\ref{ass:naive}(ii) excludes learners at or beyond the parametric rate, and this exclusion is essential rather than technical. When \(\sqrt{n}\, \delta_{n,k} = O_P(1)\), all three terms of the exact identity \eqref{eq:naive_mean} are of the same order \(n^{-1}\) as the statistic's conditional standard deviation \(s_k/\sqrt{n_k}\), so the standardized location is the \(O_P(1)\) random quantity \(\sqrt{n_k}\, \big(-\delta_{n,k}^2 + \Delta_{c,k}^2 + 2\rho_k\big)/s_k\), whose sign is not determined by \(H_0\): the limit of \(Z_{\mathrm{naive}}\) is a normal location \emph{mixture} driven by the training-fold randomness in \(\Delta_{c,k}\) and \(\rho_k\), not \(\mathcal{N}(-\kappa, 1)\) for a constant \(\kappa\). Conservativeness can then fail. Two mechanisms illustrate this. First, if \(\hat{\tau}_k \equiv \tau_{\mathrm{ATE}}\) exactly while \(\hat{\tau}_{\mathrm{ATE},k}\) is a regular AIPW estimate, then \(m_k = \Delta_{c,k}^2 > 0\) and the standardized statistic converges to \(|W|/(2\bar{\sigma}) + \mathcal{N}(0,1)\)-type limits with \(W\) Gaussian, so the one-sided test over-rejects mildly. Second, without Assumption~\ref{ass:naive}(iii) the cross term \(2\rho_k\) can be made positive and of larger order than \(\delta_{n,k}^2\) by nuisance-estimator sequences that satisfy Assumptions~\ref{ass:unconfound}--\ref{ass:nuisance_regularity} (correlated, spiked errors of \(\hat{e}_k\), \(\hat{\mu}_k\) and \(\hat{\tau}_k\) on a common small region), in which case \(Z_{\mathrm{naive}} \xrightarrow{p} +\infty\) and the benchmark over-rejects with probability tending to one. Both mechanisms are artifacts of evaluating the two losses on the same observations; neither arises for Algorithm~\ref{alg:ss_cvt}.
\end{remark}

\begin{remark}[Reading Table~\ref{tbl:sim}, and the contrast with Algorithm~\ref{alg:ss_cvt}]
  \label{rem:contrast}
  For regularized machine learners under the null, the \(L_2\) rate typically satisfies \(\sqrt{n}\, d_n \to \infty\), if only logarithmically (e.g.\ \(d_n \asymp \sqrt{s \log p / n}\) for \(\ell_1\)-regularized DR-learners), so Proposition~\ref{prp:naive}(iv) applies: \(Z_{\mathrm{naive}}\) drifts to \(-\infty\) at the slow rate \(\sqrt{n}\, d_n\) and the rejection probability decays to zero correspondingly slowly. This is consistent with the small and nearly \(n\)-insensitive rejection rates of the Naive DML column of Table~\ref{tbl:sim} over the moderate range \(n \in [250, 2000]\).

  The contrast with the Intra-Fold Sample-Split is instructive, and it lies entirely in the denominator. Under \(H_0\) and Assumptions~\ref{ass:unconfound}--\ref{ass:nuisance_regularity} and~\ref{ass:naive}, Lemma~\ref{lem:A1} shows that the split estimator carries the \emph{same} leading drift in its conditional mean, \(\hat{\theta}_{\mathrm{split}} = -d_n^2(1 + o_P(1)) + O_P(n^{-1/2})\)-fluctuations, since the dominant remainder in \(B_k\) is \(-P_0[\Delta \tau^2 \mid \mathcal{I}_{-k}]\). But by Lemma~\ref{lem:A3} and Assumption~\ref{ass:nondegen}(ii) the split statistic's standard error is of exact order \(n^{-1/2}\) and bounded below on the boundary, so the standardized drift is \(-\sqrt{n}\, d_n^2 / \sqrt{2\sigma_{\mathrm{tot}}^2 + 2\sigma_{\mathrm{res}}^2}\), which vanishes because \(d_n = o(n^{-1/4})\) under Assumption~\ref{ass:nuisance_regularity}(i); hence \(Z_{\mathrm{split}} \xrightarrow{d} \mathcal{N}(0,1)\). The benchmark instead normalizes by a standard error of order \(d_n\, n^{-1/2}\) (Proposition~\ref{prp:naive}(iii)), so the same drift is inflated by the factor \(1/d_n\) and the statistic degenerates: \(d_n^2 \big/ \big(d_n/\sqrt{n}\big) = \sqrt{n}\, d_n \to \infty\). The surviving factor of \(d_n\) in the denominator is exactly what the Intra-Fold split removes, converting a degenerate, drifting statistic into an asymptotically pivotal one.
\end{remark}

\end{document}